\documentclass[11pt]{amsart}
\usepackage{graphicx} 
\graphicspath{Images/}
\usepackage{subfiles}
\usepackage{caption}
\usepackage{geometry}
\usepackage{subcaption}
\usepackage[hidelinks]{hyperref}
\usepackage{amssymb}
\usepackage{dsfont}
\def\acts{\curvearrowright}
\usepackage{amsmath}
\usepackage{amsthm}
\usepackage{thmtools}
\usepackage{thm-restate}
\renewcommand{\dot}[0]{ {\cdot}}
\usepackage{etoolbox}
\usepackage{bbm}
\usepackage{lipsum}
\usepackage{color}

\usepackage[style=alphabetic,sorting=nty,
    sortcites=true,doi=false,url=false,
    giveninits=true,hyperref]{biblatex}
\DeclareFieldFormat[incollection]{pages}{#1}
\DefineBibliographyStrings{english}{%
  page             = {},
  pages            = {},
}

\renewbibmacro{in:}{}
\renewbibmacro*{journal}{%
  \iffieldundef{shortjournal}
    {%
      \iffieldundef{journaltitle}
        {}
        {%
          \printtext[journaltitle]
            {%
              \printfield[titlecase]{journaltitle}%
              \setunit{\subtitlepunct}%
              \printfield[titlecase]{journalsubtitle}%
             }%
         }%
    }
    {\printtext[journaltitle]{\printfield[titlecase]{shortjournal}}}%
}

\DeclareFieldFormat[misc]{title}{#1}
\DeclareFieldFormat[thesis]{title}{#1}
\DeclareFieldFormat[incollection]{title}{#1}
\DeclareFieldFormat[inproceedings]{title}{#1}
\DeclareFieldFormat[article]{title}{#1}
\DeclareFieldFormat[article]{volume}{\mkbibbold{#1}}

\DeclareFieldFormat[article]{number}{(#1)}
\DeclareFieldFormat[inproceedings]{year}{(#1)}
\renewbibmacro*{volume+number+eid}{%
  \printfield{volume}%
  \printfield{number}%
  \setunit{\addcomma\space}%
  \printfield{eid}}

\renewbibmacro*{journal+issuetitle}{%
  \usebibmacro{journal}%
  \setunit*{\addspace}%
  \iffieldundef{journaltitle}{}{\usebibmacro{volume+number+eid}}%
  \setunit{\addcolon\space}%
  \usebibmacro{issue}%
  \newunit}

\DeclareFieldFormat[article]{pages}{#1} 
\renewbibmacro*{note+pages}{%
  \printfield{note}%
  \setunit{\addcomma\space}%
  \printfield{pages}%
  \setunit{\addspace}%
  \printtext{\mkbibparens{\thefield{year}}}
  \newunit}

\usepackage[toc,page]{appendix}

\newcommand{\p}[2]{\frac{\partial^{#1}}{\partial {#2}^{#1}}}

\newcommand{\bb}[1]{\mathbb #1}
\renewcommand{\d}[0]{\; {\rm d}}
\newcommand{\td}[0]{{\rm d}}
\newcommand{\1}[0]{\mathds 1}
\newcommand{\Z}[0]{\bb Z}

\newcommand{\R}[0]{\bb R}
\newcommand{\C}[0]{\bb C}

\renewcommand{\H}[0]{\bb H}
\renewcommand{\P}[0]{\bb P}

\newcommand{\Eg}[0]{\bb E_{\Omega_g}}

\newcommand{\cF}[0]{\mathcal F}
\newcommand{\cA}[0]{\mathcal A}

\newcommand{\D}[0]{{\rm D}}
\newcommand{\mysubsec}[1]{\subsection*{\;\; #1}}
\newcommand{\od}[0]{{\rm OD}}

\newcommand{\cS}[0]{\mathcal S}

\newcommand{\cM}[0]{\mathcal M}
\newcommand{\cN}[0]{\mathcal N}
\newcommand{\cMg}[0]{{\mathcal M_g}}
\newcommand{\cMgn}[0]{\mathcal M_{g,n}}
\newcommand{\cT}[0]{\mathcal T}

\newcommand{\cP}[0]{\mathcal P}
\newcommand{\xz}[0]{{(X,z)}}

\newcommand{\cL}[0]{\mathcal L}

\newcommand{\G}[0]{\Gamma}
\newcommand{\g}[0]{\gamma}
\newcommand{\ug}{T}
\newcommand{\simf}[0]{\overset{free}{\sim}}

\newcommand{\supp}[0]{{\rm supp\;}}
\newcommand{\hf}[0]{{\hat f}}
\newcommand{\hh}[0]{\hat h}

\newcommand{\id}[0]{{\rm id}}

\newcommand{\ra}[0]{\rightarrow}

\newcommand{\ba}[0]{\backslash}

\newcommand{\WP}[0]{{\rm WP}}

\newcommand{\psltr}[0]{{\rm PSL}_2(\R)}
\newcommand{\mcg}[0]{{\rm MCG}}
\newcommand{\mcgg}[0]{{\rm MCG}_g}

\newcommand{\Var}[0]{{\rm Var}}

\newcommand{\Varg}[0]{\Var_{\Omega_g}}

\newcommand{\cG}[0]{\mathcal G}

\newcommand{\sns}[0]{{\rm sns}}
\newcommand{\ssep}[0]{{\rm ssep}}
\newcommand{\rest}[0]{{\rm rest}}
\newcommand{\ns}[0]{{\rm nsim}}

\newcommand{\ir}[0]{{\rm InjRad}}

\newcommand{\ach}[0]{{\rm arccosh}}
\newcommand{\ash}[0]{{\rm arcsinh}}

\newcommand{\csch}[0]{{\rm csch}}
\newcommand{\eps}[0]{\varepsilon}
\newcommand{\imii}[0]{\int_{-\infty}^\infty}

\newcommand{\izi}[0]{\int_0^\infty}
\newcommand{\E}[0]{\mathbb E}
\newcommand{\prt}[1]{\left( #1 \right)}
\newcommand{\sqb}[1]{\left[ #1 \right]}

\numberwithin{equation}{section} 
\newtheorem{theorem}{Theorem}[section]
\newtheorem{lemma}[theorem]{Lemma} 
\newtheorem*{lemma*}{Lemma}

\newtheorem{proposition}[theorem]{Proposition}
\newtheorem{corollary}[theorem]{Corollary}
\newtheorem{claim}{Claim}
\newtheorem*{claim*}{Claim}
\theoremstyle{definition}
\newtheorem{numnotation}{Notation}

\numberwithin{numremark}{section}
\theoremstyle{remark}
\newtheorem*{remark}{Remark}

\theoremstyle{definition}
\newtheorem{definition}[theorem]{Definition}
\theoremstyle{plain}
\numberwithin{theorem}{section}

\newcommand{\osc}[0]{{\rm osc}}

\newcommand{\vol}[0]{{\rm vol}}
\newcommand{\area}[0]{{\rm area}}

\newenvironment{proofblack}[1][Proof]{%
  \begin{proof}[#1]%
}{%
  \end{proof}%
}

\renewcommand{\paragraph}[1]{%
  \;\newline\textbf{#1}%
}
\renewcommand{\phi}{\varphi}
\title[Local Weyl law and length-minimising loops]{Local Weyl law and length-minimising loops\\ on hyperbolic surfaces}
\author{Daniel Meriaz}
\address[Daniel Meriaz]{School of Mathematics, University of Bristol, Bristol, BS8 1UG, U.K.}
\email{daniel.meriaz@bristol.ac.uk}

\begin{document}
    \begin{abstract}
        We study the variance of a local Weyl law over a fixed smooth energy window, when averaged over large Weil--Petersson hyperbolic surfaces. 
        Our results are consistent with the predictions of Berry's random wave model.
        Our approach allows to explicitly integrate certain test functions which depend on lengths of based geodesic loops, and relate them to the associated lengths of the closed geodesics in their free-homotopy class. 
        We thus utilise Mirzakhani’s work, with exact stationary phase arguments, to identify correct main and error terms, making explicit the asymptotic behaviour of the variance of the local Weyl law. Furthermore, we introduce the geometric notion of length-minimising geodesic loops and sequences, based at a point. We prove a complete characterisation of the topology of these, namely that they are simple. This forms a key ingredient in our study, and yields a new streamlined argument to bound the contributions of remainder terms which depend on lengths of pairs of different short primitive geodesic loops. To illustrate the generality of our results, we further introduce a family of ``exploring" loops based at a point, which might be of independent interest.
    \end{abstract}
\maketitle
\tableofcontents
\section{Introduction}\label{sec: introduction}
Let $X$ be a compact, oriented, and connected hyperbolic surface without boundary. Let $\{\phi_j \}_{j\geq 0}$ be an orthonormal eigenbasis of Laplacian eigenfunctions, with associated eigenvalues $0=\lambda_0 < \lambda_1 \leq \lambda_2 \leq \cdots \to \infty$. Furthermore, let $r_j :=\sqrt{\lambda_j - \frac 14}\in \R_{\geq 0} \cup i(0, \frac 12]$.

We study the \emph{local Weyl law} of eigenfunctions, over a fixed energy window, in the regime of large-volume Weil--Petersson random hyperbolic surfaces. 

To fix notation, let $f$ be an even test function  with compactly supported Fourier transform $\hf\in C_c^\infty(\R)$,
and pick two parameters $L\geq1,\tau>0$. Define $h(r):= f(L(r-\tau))+f(L(r+\tau))$.
We analyse the following function, defined at points $z\in X$,
\begin{equation*}
    N_{f,L,\tau}\xz :=N\xz := \sum_{j\geq 0} h(r_j) |\phi_j(z)|^2.
\end{equation*}
This is a smoothed version of the probability densities of eigenfunctions in a frequency window of width $1/L$ about the frequency $\tau$.

We study the expectation and variance of $N\xz$ over the following space of \emph{hyperbolic surfaces with a marked point},
\begin{equation*}
    \Omega_g:= \{\xz:\text{$X$ is a genus $g$ hyperbolic surface, }z\in X \}/ \text{marked isometries},
\end{equation*}
where $X$ is random with respect to the normalised Weil--Petersson measure, and $z$ is uniformly distributed over $X$. See Section \ref{subsec: the probability space} for precise definitions.

The goal of this paper is to prove the following asymptotic on the \emph{variance} of $N\xz$.
\begin{theorem}\label{thm: main theorem on the variance}
    For fixed $\tau>0, f$, and $L\geq 1$, as $g\to\infty$,
    \begin{align*}
         \Varg \prt{N\xz} &= \frac{1}{4\pi(g-1)}\frac{\tau \tanh \pi \tau}{\pi L} \|f\|_{L^2(\R)}^2 + O_{f,\tau}\prt{\frac{1}{L^2g}} 
        +O_{f,L,\tau}\prt{\frac{1}{g^2}}.
    \end{align*}
    Moreover,
    \begin{equation*}
        \Eg \sqb{N\xz} = \frac{1}{4\pi}\imii h(r) r\tanh\pi r \d r + O_{f,L,\tau}\prt{\frac 1g}.
    \end{equation*}
\end{theorem} 
\mysubsec{Gaussian eigenfunction behaviour}
It is expected, in the high energy limit of \emph{generic} chaotic systems, that \emph{all statistics of spatial fluctuations of eigenfunctions} are predicted by \emph{Berry's random wave model} (RWM) \cite{MVBerry_1977}. This RWM is generated by a superposition of Gaussian waves of fixed frequency and random phases.
A certain corollary of this model is that high-energy eigenfunctions should behave like suitably normalised i.i.d.\ Gaussians, independent of their energies (c.f.\ also the \emph{Gaussian distribution conjecture} \cite{HejhalRacknertopographyofmaasswaveofrms} \cite{humphriesEquidistributionShrinkingSets2017}).

We prove in Appendix \ref{app: berry agrees} (c.f.\ Theorem \ref{thm: random wave}) that Theorem \ref{thm: main theorem on the variance} agrees with this prediction.
Explicitly, we model the eigenfunctions $\{\phi_j \}_{j=0}^\infty$ by a sequence $\{\xi_j^X\}_{j=0}^\infty$ of i.i.d.\ real centred Gaussians, of variance $(4\pi(g-1))^{-1}$ on $X$. Then,
\begin{equation*}
    \Varg \prt{\sum_{j\geq 0} h(r_j) |\xi_j^X(z)|^2} = \Varg \prt{N\xz}+ E,
\end{equation*}
where $E=E(f,L,\tau,g)$ is bounded by the same error terms appearing in Theorem \ref{thm: main theorem on the variance}.
\mysubsec{Strategy}
To prove the asymptotics of Theorem \ref{thm: main theorem on the variance},  we use Selberg's pre-trace formula to express our local Weyl law as
\begin{equation*}
    N\xz = \bar N + N_\osc \xz
\end{equation*}
where $\bar N$ is the constant term
\begin{equation*}
    \bar N := \imii h(r) r\tanh\pi r\d r,
\end{equation*} 
and $N_\osc\xz$ is a certain (oscillatory) sum over lengths of \emph{geodesic loops} based at $z$, which depends on the choice of $f,L,\tau$. 

For our purposes of computing the variance, it suffices to analyse the second moment
\begin{equation*}
    \Eg[N_\osc \xz^2].
\end{equation*}
The random variable $N_\osc\xz^2$ is a sum over \emph{pairs of geodesic loops} based at $z$.

By our particular choice of random variable, the double sum restricts to a sum over pairs of geodesic loops of bounded length, depending on $L$ and $f$. We call such geodesic loops \emph{short}.

We split the double sum $N_\osc\xz^2$ into two: pairs of geodesic loops that are powers of the same primitive loop, and pairs that are not. These we call the \emph{diagonal} and \emph{off-diagonal} terms, respectively. See Figure \ref{fig: diagonal and off-diagonal example} for an example of this split.
\begin{figure}[h!]
    \centering
    \begin{subfigure}[t]{0.45\textwidth}
        \centering
        \includegraphics[width=0.95\linewidth]{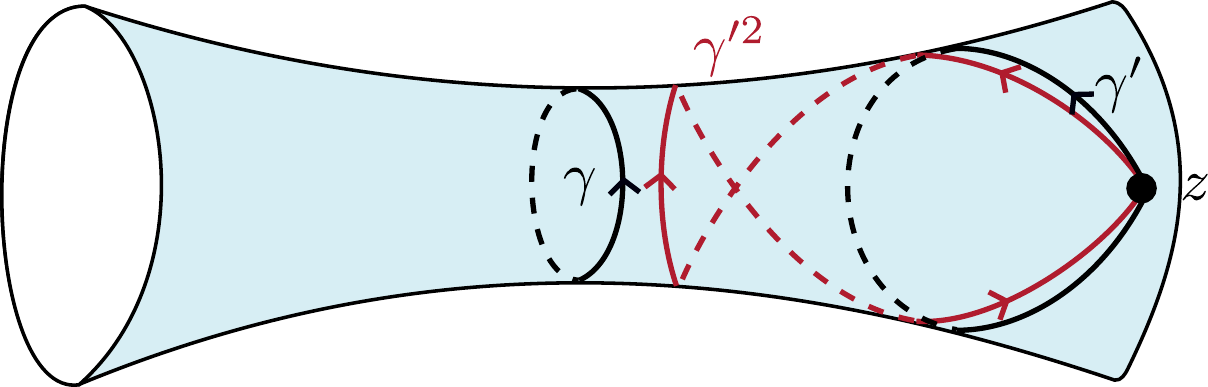}
        \caption{Pair of geodesic loops $\g', \g'^2$ in the diagonal terms.}
        \label{fig: cylinderdiagonalpowers}
    \end{subfigure}
    ~
    \begin{subfigure}[t]{0.4\textwidth}
        \centering
        \includegraphics[width=0.95\linewidth]{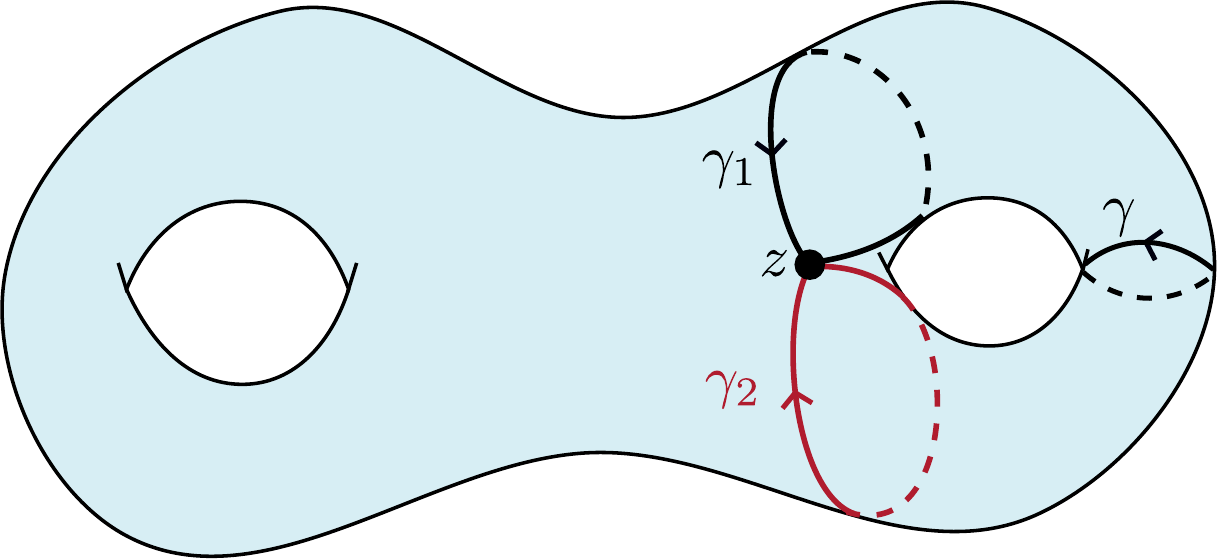}
        \caption{Pair of (primitive) geodesic loops $\g_1,\g_2$ in the off-diagonal terms.}
        \label{fig: offdiagexample}
    \end{subfigure}
    \caption{Example of diagonal and off-diagonal pairs of geodesic loops based at a point $z$. They are freely homotopic to (powers of) the closed geodesic $\g$.}
    \label{fig: diagonal and off-diagonal example}
\end{figure}
\mysubsec{Diagonal terms}
The main term of Theorem \ref{thm: main theorem on the variance} arises from the diagonal terms, whereas the off-diagonal terms yield contributions to the $O(g^{-2})$ remainder terms.

In the diagonal terms, on any fixed surface $X$, one first groups the pairs by the associated \emph{closed geodesic} the primitive geodesic loop is freely-homotopic to. Then, in a computation not too dissimilar to computations in the proof of Selberg's trace formula, the space average of the diagonal terms reduces to a sum over the length of this {closed geodesic}, forgetting about the geodesic loop (c.f.\ Lemma \ref{lem: general space average} and Proposition \ref{prop: diagonal on deterministic}).

Thanks to Mirzakhani's pioneering work (c.f.\ Section \ref{subsec: mirz int formula}), we can explicitly compute this expectation. After a stationary phase argument, the main term in Theorem \ref{thm: main theorem on the variance} arises out of the contributions of the simple nonseparating closed geodesics, $\g_1 = \g_2^{\pm1}$. 

This is the subject of Section \ref{sec: diagterms}.
\mysubsec{Length-minimising loops}
We introduce the new notion of \emph{length-minimising loops} based at a point $z\in X$, c.f.\ Definition \ref{def: length min loop}. Informally, a length-minimising loop  with respect to some subgroup $G\lneq \pi_1\xz$, is a geodesic loop whose homotopy class is not in $G$, with minimal length among such geodesic loops. For example, a systole $\g_1$ is length-minimising with respect to $G_1 := \{\id\}$; and the \emph{second shortest primitive loop}, which we here define, is length-minimising with respect to $G_2:= \{\g_1^m \}_{m\in \Z}$. 

As a corollary of more general results on \emph{length-minimising sequences} (c.f.\ Corollary \ref{cor: length-minimising topology}), we prove the following important result on the shape of the second shortest primitive loop, c.f.\ Figure \ref{fig: second prim is simple}.
\begin{theorem}\label{thm: second shortest simp and triv}
    Let $\g_2$ be a second shortest primitive loop, based at $z$, with respect to a systole $\g_1$. Then $\g_2$ is simple and intersects $\g_1$ trivially at $z$.
\end{theorem}
\begin{figure}[h!]
    \centering
        \includegraphics[width=0.25\textwidth]{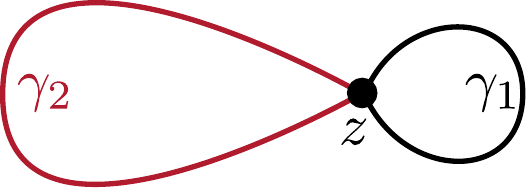}
    \caption{The second shortest primitive loop $\g_2$ is simple and intersects the systole $\g_1$ trivially.}
    \label{fig: second prim is simple}
\end{figure}

As an illustration of the generality of the definitions of length-minimising loops and sequences, we introduce the family of \emph{exploring loops} (c.f.\ Definition \ref{def: exploring loops}), which might be of independent interest. These are a length-minimising sequence $\g_1,\ldots,\g_n$ of geodesic loops, where $n\leq 2g-1$. They  iteratively ``explore" the topology of the surface, trivially intersect, and each of these are simple. By ``explore" we mean that the surface weakly filled by $\g_1,\ldots,\g_j$ has Euler characteristic $1-j$, for all $j\in \{1,\ldots, n\}$ (c.f.\ Theorem \ref{thm: euler char of exploring}). Thus, we exhibit a family of iteratively defined length-minimising geodesic loops, whose topology we know, which weakly fill the surface.

This is the subject of Section \ref{sec: len minimising loops}.
\mysubsec{Off-diagonal terms}
We combine our previous results to bound the size of the off-diagonal terms, which, as noted above, correspond to lengths of pairs of geodesic loops based at $z$, which are powers of two different primitive loops.

For the off-diagonal terms to be nonzero at a point $z$, there need to be \emph{two different} short primitive geodesic loops based at $z$, by our choice of test function. By length-minimality, this means that \emph{both} the systole and the second shortest primitive loop need be short.

We can then dispense of the expectation of the off-diagonal terms, and send it to the remainder terms. Thanks to Theorem \ref{thm: second shortest simp and triv}, and work of Monk and Thomas \cite{monkthomastanglefree}, by relating the short pair of loops $\g_1,\g_2$ to the short embedded geodesic once-holed torus or pair of pants that they fill, we prove that the off-diagonal terms yield an error term of $O_{f,L,\tau}\prt{g^{-2}}$ to the variance (c.f.\ Proposition \ref{prop: off diagonal term claim}).

We would like to emphasize a certain aspect of this proof. To relate these points to the short geodesic pair of pants or once-holed tori they fill, a dependency on the injectivity radius at the point $z$ arises. We use the fact that one has two short loops to our advantage. Using a certain \emph{loop collar lemma}, Lemma \ref{lem: l_ loop collar lemma}, the injectivity radius is uniformly lower bounded by a function of $L$.

Thus, points with two different primitive short loops, cannot have small injectivity radius. 

This is the subject of Section \ref{sec: offdiagterms}.
\mysubsec{Related works}
The results of this paper sit firmly in the field of \emph{quantum chaos} (c.f.\ the survey paper \cite{marklof2006arithmetic}).  Hyperbolic surfaces are a good toy model in this field, owing to the chaotic properties of their geodesic flow, and the exact pre-trace formula of Selberg's \cite{Selberg_1956} (c.f.\ Lemma \ref{lem: selberg trace formula}), relating spectral data of a surface with its geometry.

Recently, through the work of Mirzakhani (c.f.\ the survey of her work \cite{wright2020tour}), it has become possible to probabilistically study (with respect to Weil--Petersson probability) spectral properties of hyperbolic surfaces, c.f.\ the survey article \cite{monknaudsurvey}.

Rudnick in \cite{rudnick2023goestatisticsmodulispace} proved that the variance of eigenvalues in a fixed-energy window, when averaged over moduli space, behave as postulated by the Bohigas--Giannoni--Schmit conjecture \cite{BGSconjecture}. These results were extended by Rudnick--Wigman \cite{RudnickWigmanCLT} \cite{rudnick2025almost} to prove central limit and almost sure convergence results. Their results may be viewed in terms of the space average $\int_X N\xz \d z$ of our local Weyl law $N\xz$, whereas we directly analyse eigenfunctions pointwise.

A central conjecture in the field is the \emph{quantum unique ergodicity} conjecture of Rudnick and Sarnak \cite{rudnicksarnakQUE}, which implies complete delocalisation of high-energy eigenfunctions. Lindenstrauss \cite{lindenstrauss2006invariant} and Soundararajan \cite{soundararajan2010quantum} proved this conjecture in certain arithmetic cases of hyperbolic surfaces. 

In the spirit of large-scale eigenfunction delocalisation results, Brooks and Lindenstrauss \cite{brooks_non-localization_2013} proved that individual eigenvectors on large regular graphs cannot concentrate a fixed proportion of their $L^2$-mass on small vertex sets. Anantharaman and Le Masson \cite{anantharaman2015quantum} strengthen this by proving quantum ergodicity for most eigenvectors on large regular graphs with few short cycles. 

In the setting of hyperbolic surfaces, Le Masson and Sahlsten \cite{lemassonsahlsten2017} establish an analogous fixed-window quantum ergodicity theorem under Benjamini–Schramm convergence. In Gilmore--Le Masson--Sahlsten--Thomas \cite{GLMST2021} various $L^p$-norm bounds on eigenfunctions of large surfaces are proven. For their results, they too required an understanding of the distribution of points with several short primitive geodesic loops, though the tools and results in this paper are different from theirs, as we take averages over the points of the surfaces themselves.
Thomas \cite{thomas_delocalisation_2022} obtained quantitative non-concentration bounds for individual eigenfunctions, similarly to Brooks--Lindenstrauss, which imply results in high probability for large-genus Weil--Petersson random surfaces. More recently, Hippi \cite{hippi2026quantummixingbenjaminischrammconvergence} extended quantum ergodicity to quantum mixing, analysing off-diagonal elements. Hippi--Lequen--Mikkelsen--Sahlsten--Uebersch\"ar \cite{hippi2026quantummixingschrodingereigenfunctions} generalise this result from Laplace eigenfunctions to eigenfunctions of Schrödinger operators with potentials.

For local Weyl laws in a fixed deterministic setting,  Karnaukh \cite{karnaukh1996thesis} and Jakobson--Polterovich \cite{jakobsonEstimatesSpectralFunction2007} prove lower bounds on the remainder terms of the local Weyl law on negatively curved manifolds.  

We finish by noting the works of Abert--Bergeron--Le Masson \cite{abert:hal-03918418} and Garc\'ia-Ruiz \cite{garciaruiz2023relationdifferentformulationsberrys} which state rigorous conjectures for high-energy eigenfunctions to have the Berry random wave property.
\mysubsec{Plan of the paper}
Besides the introduction, the article consists of five other sections:
\begin{itemize}
    \item[(i)] Section \ref{sec: prelims and defs}: provides some preliminaries on the analysis of hyperbolic surfaces and their moduli spaces.
    \item[(ii)] Section \ref{sec: eigenfunction stat}: introduces the local Weyl law and the probability space we will work with, and proceed with the first steps towards the proof of Theorem \ref{thm: main theorem on the variance}.
    \item[(iii)] Section \ref{sec: diagterms}: calculates the asymptotics of the diagonal terms.
    \item[(iv)] Section \ref{sec: len minimising loops}: defines length-minimising loops and sequences, and prove their simplicity. Furthermore, we provide the definition and prove properties of exploring loops.
    \item[(v)] Section \ref{sec: offdiagterms}: estimates the contributions of the off-diagonal terms, using Theorem \ref{thm: second shortest simp and triv} and the results in Section \ref{sec: len minimising loops}.
\end{itemize}

\mysubsec{Acknowledgements}
I would like to express my gratitude to my advisors, Jens Marklof and Laura Monk, for suggesting the problem, and for their invaluable discussions and support throughout. I am grateful to Yuxin He for an enlightening observation made on an earlier version of a result in this paper, and Ze\'ev Rudnick for helpful comments on a previous draft of the paper. I would also like to thank Yulin Gong, Aapo Karvonen, and Luke Turvey for many interesting conversations.

This research was supported by EPSRC DTP S101056-120.

\mysubsec{Data access statement} No new data were generated nor analysed during this study.

\section{Preliminaries and definitions}\label{sec: prelims and defs}
We will use the following asymptotic notation throughout: we say $f(x) = O_A(g(x))$ or $f(x) \ll_A g(x)$ if there exists a constant $C = C(A) > 0$ such that $|f(x)| \leq C|g(x)|$ for all $x>1$.
\subsection{Topological definitions}\label{subsec: topological definitions}
    We define the topological notions that we will use throughout. These are mainly concerned with the notion of loops and fundamental group.
    
        Let $X$ be a connected topological space. An \textit{arc} is a continuous function $c:[0,1]\ra X$ with \emph{endpoints} $e_0(c):=c(0), e_1(c):=c(1)$. We say an arc is a \emph{loop} if $e_0(c)=e_1(c)$ and call $e_0(c)$ its \emph{basepoint}.

        A \textit{self-intersection} is defined as two times $0 \leq t_1 < t_2 < 1$ such that $c(t_1) =c(t_2)$, with \textit{intersection point} $c(t_1)=c(t_2)$. An arc is \textit{simple} if it has no self-intersections. Two arcs sharing at least one endpoint \emph{intersect trivially} if they intersect each other only at times $0$ and $1$.

        We say two arcs $c_0,c_1$ are \textit{compatible} if $e_1(c_0) = e_0(c_1)$. Loops with the same basepoint are compatible. 
        
        For two compatible arcs $c_0,c_1$ we define their \textit{concatenation} as
        \begin{equation*}
            (c_0\dot c_1)(t):= \begin{cases}
                c_0(2t) & t\leq \frac 12 \\
                c_1(2t-1) & t>\frac 12.
            \end{cases}
        \end{equation*}
        We furthermore let $\bar c(t):= c(1-t)$ denote the \textit{reversal} of $c$.
        
        Two loops $c_0,c_1:[0,1]\ra X$ based at $z$ are \textit{homotopic}, which we denote by $c_0\sim c_1$, if there exists a \textit{homotopy} between them, i.e.\ a continuous function $h:[0,1]\times [0,1] \ra X$ such that $h(0,\cdot) = c_0, h(1,\cdot)=c_1$ and $\forall t\in [0,1], h(t,\cdot)$ is a loop based at $z$.
        The \textit{fundamental group based at} $z, \pi_1\xz$, is the set of homotopy classes of loops based at $z$, where the group multiplication is defined as concatenation of homotopy class representatives, and inverses are reversals. As $X$ is connected, up to isomorphism, the group $\pi_1\xz$ is independent of the basepoint $z$, and we denote the \emph{fundamental group} by $\pi_1(X)$.

        We say a loop is \textit{contractible} if it is homotopic to the loop $c(t) \equiv c(0)$, i.e.\ it is in the homotopy class of the identity.
        
        We say a loop is \textit{primitive} if it is not homotopic to a proper power of another loop. Our definition precludes the identity loop from being primitive.

        Two loops $c_0,c_1:[0,1]\ra X$ are \emph{freely homotopic}, denoted by $c_1 \simf c_2$, if there exists a free-homotopy between them, i.e.\ a continuous function $h:[0,1]\times [0,1] \ra X$ such that $h(0,\cdot) = c_0, h(1,\cdot)=c_1$ and $h(\cdot, 0)=h(\cdot,1)$.

        Finally, if $(X,d_X)$ is a metric space, then we denote the \emph{length} of an arc $c$ by $\ell(c)$, i.e.\ 
        \begin{equation*}
            \ell(c) := \sup_{0=t_0<t_1<\cdots< t_n =1} \left\{\sum_{i=1}^n d_X(c(t_{i-1}),c(t_{i})) \right\}.
        \end{equation*}
\subsection{Hyperbolic space and hyperbolic surfaces}
We will work with the Poincar\'e upper half-plane 
\begin{equation*}
    \H = \{z= x+iy: x,y\in \R, y>0 \},
\end{equation*}
as a model for hyperbolic space.
It is endowed with the standard hyperbolic metric
\begin{equation*}
    \d s^2=\frac{\d x^2+{\rm d}y^2}{y^2}.
\end{equation*}
The distance between $z_1,z_2\in \H$ with respect to the above metric is denoted by $d(z_1,z_2)$. 

We have the hyperbolic volume form
\begin{equation*}
    \td z=\frac{\td x\td y}{y^2}.
\end{equation*}

The Laplace--Beltrami operator (or simply the Laplacian) is
\begin{equation*}
    \Delta = -y^2\prt{\p{2}{x}+\p{2}{y}}.
\end{equation*}

Throughout the article we will denote by $X$ a closed, connected, and oriented hyperbolic surface.

Such a surface has $\H$ as its universal cover, and it can be identified with $\G \ba \H$ where $\G \leq \psltr \cong \text{Isom}^+(\H)$ is a cocompact Fuchsian group such that $\G \cong \pi_1(X)$. The metric on $X$ is naturally inherited from $\H$ by the quotient of this action.

\subsection{Geodesic loops and closed geodesics}
From the hyperbolicity of $X$, every homotopy class of loops in $\pi_1\xz$ has a unique length-minimising representative, which we will call a \emph{geodesic loop} (c.f.\ \cite[Theorem 1.5.3]{buser1992geometry}).

Thus $\pi_1 \xz$ is in one-to-one correspondence with the set of geodesic loops based at $z$. We will generally identify between the geodesic loop representative and their homotopy class interchangeably, provided the basepoint $z$ is implicitly understood.

If $c$ is a loop based at $z$, then we denote by $\ell_\xz(c)$ the length of the unique geodesic loop in the homotopy class of $c$. If $c$ is contractible then we set $\ell_\xz (\g) := 0$. 

Equivalently, if $\g\in \pi_1\xz$ is a geodesic loop, then $\ell(\g) = d(\tilde z,\tilde\g.\tilde z)$ where $\tilde z,\tilde \g$ are lifts to the universal cover $\H$.

Every loop on $X$ is \emph{freely homotopic} to a unique length-minimising free-homotopy representative which we call a \emph{closed geodesic} (c.f.\ \cite[Theorem 1.6.6]{buser1992geometry}). These are in one-to-one correspondence with the \emph{conjugacy classes} of $\pi_1(X)$. If $\g$ is a closed geodesic, we let $\ell(\g)$ be its length. 
Equivalently $\ell(\g) = d(\tilde z,\tilde \g.\tilde z)$, where $z\in \g$, and $\tilde z,\tilde \g\in \psltr$ are lifts to the universal cover $\H$.

We let $\cP\subseteq\pi_1(X)$ be the set of primitive elements of the fundamental group, and we denote the conjugacy classes of primitive elements by $[ \cP ]$. On hyperbolic surfaces, every element $\id\neq\gamma\in \pi_1(X)$ has a unique primitive element $\gamma_0\in \cP$ and unique $n\geq 1$ such that $\gamma = \gamma_0^n$.

We conclude by recalling a few standard facts about geodesic loops.
\begin{lemma}\label{lem: on geodesic loops}\;
        \begin{itemize}
            \item[{(\rm i)}]  A geodesic loop is smooth outside of its endpoints.
            \item[{\rm (ii)}]  A geodesic loop's self-intersections are finite and transverse, except perhaps at its basepoint.
            \item[{\rm (iii)}]  If $c$ is a loop that is not a geodesic loop, then $\ell_\xz(c) < \ell(c)$.
        \end{itemize}
    \end{lemma}
\subsection{Integral transforms}
Throughout we will make use of the following transforms, for functions $f$ in an appropriate function space.

We use the non-normalised \emph{Fourier transform} 
\begin{equation}\label{def: fourier transform}
        \mathcal Ff(x)=\hat f(x):= \frac{1}{2\pi}\int_{-\infty}^\infty f(y)e^{-ixy}\d y,
\end{equation}
whose \emph{inverse Fourier transform} is
\begin{equation*}
    \cF^{-1}f(x)=\int_{-\infty}^\infty f(y)e^{ixy}\d y.
\end{equation*}

Now let $f\in C_c^\infty(\R)$. We define the \emph{Abel transform}
\begin{equation}\label{def: A and Qbel transform}
        \cA f(u)=\sqrt 2 \int_{u}^\infty \frac{f(\rho)\sinh \rho}{\sqrt{\cosh\rho-\cosh u}}\d\rho,
\end{equation}
with inverse given by the formula 
\begin{equation*}
    \cA^{-1} f(\rho)=-\frac{1}{\sqrt 2 \pi}\int_\rho^\infty \frac{f'(u)}{\sqrt{\cosh u- \cosh\rho}}\d u.
\end{equation*}

Finally, we define the \emph{Selberg transform} as 
    \begin{equation}\label{def: selberg transform}
        \cS:= \cF^{-1}\circ \cA.
    \end{equation}
\subsection{Selberg's (pre-)trace formula}
The following fundamental theorem allows us to relate between the spectral data and the geometry of a surface. 
\begin{theorem}[Selberg's pre-trace formula \cite{Selberg_1956}]\label{thm: pre-trace formula}
    For $h$ even with $\hh\in C_c^\infty(\R)$, define $k := \mathcal S^{-1}(h)$. Furthermore, let $\{\phi_j \}_{j\geq 0}$ be an $L^2$-orthonormal eigenbasis of the Laplacian on the surface $X$ with eigenvalues $0=\lambda_0<\lambda_1\leq\lambda_2\leq \cdots$. Let $\lambda_j = \frac 1 4 +r_j^2$ with $r_j\geq0$ or $r_j\in i(0, \frac 12]$. Then for all $z\in X$,
    \begin{equation}\label{eq:pre-trace formula}
        \sum_{j=0}^\infty h(r_j)|\phi_j(z)|^2 = \frac{1}{4\pi}\int_{-\infty}^\infty h(r)\tanh(\pi r)r \d r+\sum_{\gamma \in \pi_1\xz -\{\id\}}k(\ell_\xz(\g)).
    \end{equation}
\end{theorem}
Integrating the pre-trace formula on both sides over $X$ yields \emph{Selberg's trace formula}.
\begin{lemma}[Selberg's trace formula \cite{Selberg_1956}]\label{lem: selberg trace formula}
    Let $h$ be even with $\hh\in C_c^\infty(\R)$. Then with the same notation as above,
    \begin{equation*}
        \sum_{j=0}^\infty h(r_j) = \frac{\area(X)}{4\pi}\imii h(r)\tanh (\pi r)r\d r + \sum_{\g\in [\cP]}\sum_{n=1}^\infty \frac{\ell(\g)}{2\sinh \prt{\frac{n\ell(\g)}{2}}}\hh(n\ell(\g)).
    \end{equation*}
    In particular, the following integral identity holds 
    \begin{equation}\label{eq: integral identity stf}
        \int_X\sum_{\gamma \in \pi_1\xz -\{\id\}}k(\ell_\xz(\g))\d z = \sum_{\g\in [\cP]}\sum_{n=1}^\infty \frac{\ell(\g)}{2\sinh \prt{\frac{n\ell(\g)}{2}}}\hh(n\ell(\g)).
    \end{equation}
\end{lemma}
\subsection{Moduli spaces and the Weil--Petersson volume}\label{subsec: mod spaces and WP vol}
Throughout this section, we fix $g,n\geq 0$ with $2-2g-n<0$, and $\Sigma_{g,n}$ to be a smooth, compact, connected, and oriented surface of genus $g$ and $n$ boundary components, labelled by $\{1,\ldots, n\}$.

    For a length vector $\vec x =(x_1,\ldots, x_n)\in \R_{\geq 0}^n$ we define the \textit{moduli space} $\cM_{g,n}(\vec x)$ to be the set of bordered hyperbolic surfaces of genus $g$ with $n$ geodesic boundary components $b_1,\ldots, b_n$ such that for all $i, \ell(b_i)=x_i$, quotiented by the isometries preserving $b_i$ set-wise for all $i$. A boundary geodesic $b_i$ with $\ell(b_i) = 0$ is a cusp, by convention. 
    
Define the \emph{mapping class group} 
\begin{equation*}
    \mcg_{g,n}:=\left\{\begin{aligned}
        &\text{orientation-preserving homeomorphisms } h:\Sigma_{g,n}\to \Sigma_{g,n}  \\
        & \text{s.t.\ } h(b_i) = b_i, \forall i
    \end{aligned}
     \right\}
     /\text{isotopy}
\end{equation*}
and the space of marked hyperbolic surfaces, the \textit{Teichm\"uller space}
\begin{equation*}
    \cT_{g,n}(\vec x):= \left\{(Y,\phi):\;\begin{aligned}
        &Y\text{ bordered hyperbolic surface of signature } (g,n) \\
        &\phi:\Sigma_{g,n}\ra Y \text{ homeomorphism} \\
        & \forall i, \phi(b_i) \text{ is a boundary geodesic of length } x_i
    \end{aligned}\right\}/ \sim
\end{equation*}
where 
\begin{equation*}
    (Y,\phi)\sim (Z,\psi) \iff \exists \text{ isometry } i:Y\ra Z \text{ s.t. } \psi^{-1} \circ i \circ \phi \text{ is isotopic to } \id_{\Sigma_{g,n}}.
\end{equation*}
One has the action $\mcg\prt{\Sigma_{g,n}}\acts \cT_{g,n}(\vec x)$ by precomposition of the marking, i.e.\
\begin{equation*}
    \forall h\in \mcg_{g,n}, \quad h.(Y,\phi):= (Y,\phi\circ h^{-1}),
\end{equation*}
and we have the following quotient identification
\begin{equation*}
    \cM_{g,n}(\vec x)=\cT_{g,n}(\vec x)/\mcg_{g,n}.
\end{equation*}

The spaces $\cM_{g,n} := \cM_{g,n}(\vec 0)$ and $ \cT_{g,n}:= \cT_{g,n}(\vec 0)$ correspond to surfaces of genus $g$ with $n$ cusps, and similarly $\cMg, \cT_g, \mcgg$ to those without boundaries.

By work of Weil \cite{Weil} and Wolpert \cite{wolpertElementaryFormulaFenchelNielsen1981}, the space $\cT_{g,n}(\vec x)$ has a natural symplectic form, the \emph{Weil--Petersson form}, which is invariant under $\mcg_{g,n}$, and hence induces a volume form on $\cM_{g,n}(\vec x)$ which we denote by$\d\vol_{g,n,\vec x}^\WP$. One has that the total volume
\begin{equation*}
    \vol_{g,n,\vec x}^\WP \prt{\cMgn (\vec x)} <\infty.
\end{equation*}
We let $V_{g,n}(\vec x) := \vol_{g,n,\vec x}^\WP \prt{\cMgn (\vec x)}$ (and $V_{g,n}, V_g$ following the above conventions) with the exception of the $g=n=1$ case, where we set $V_{1,1}(x):= \frac 12\vol_{1,1, x}^\WP \prt{\cM_{1,1} ( x)}$. The constant $\frac 12$ arises from the existence of an involution symmetry on any once-holed torus with boundary (c.f.\ \cite[Section 2.8]{wright2020tour}).
\subsection{Mirzakhani's integration formula}\label{subsec: mirz int formula}
We introduce the mechanism by which we can integrate certain functions over moduli spaces.

First, we say a loop $c$ on a surface $\Sigma_{g,n}$ is \emph{essential} if it is not freely homotopic to a point nor to a boundary component (or cusp) of $\Sigma_{g,n}$. 
\begin{definition}[Multicurves]\label{def: multicurves}
A \emph{multicurve} $\vec \g$ is an ordered tuple of disjoint essential simple loops $(\g_1,\ldots, \g_k)$ such that for every index $i\neq j$, $\g_i$ is not freely homotopic to neither $\g_j$ nor $\g_j^{-1}$.
\end{definition}

The mapping class group naturally acts on free homotopy classes of multicurves by postcomposition. We say that two free homotopy classes of multicurves $\vec\g_1, \vec \g_2$ have the same \emph{topological type} if they are contained in the same mapping class group orbit. 

Any free homotopy class of multicurves $\vec \g = (\g_1,\ldots,\g_k)$ has a (unique up to reparametrization) representative which is a multi-geodesic. Thus we can define the \textit{total length} $\ell_{\Sigma_{g,n}}(\vec\g)$ and \textit{length vector} $\vec \ell_{\Sigma_{g,n}}(\vec\g)$ as
\begin{equation*}
    \ell_{\Sigma_{g,n}}(\vec\g):= \sum_{i=1}^k \ell_{\Sigma_{g,n}}(\g_i), \quad\vec\ell_{\Sigma_{g,n}}(\vec\g):= (\ell_{\Sigma_{g,n}}(\g_1),\ldots,\ell_{\Sigma_{g,n}}(\g_k)).
\end{equation*}

\begin{definition}[Geometric functions]\label{def: geometric functions}
    Let $\vec \g$ be a multicurve, and $F:\R_{>0}^k \to \R$ a compactly supported function. A \emph{geometric function} is a function $\cM_g\ra\R$ that can be written as:
    \begin{equation*}
        F^{\vec\g}(X)=\sum_{\vec\alpha \in \mcg. \vec \g} F(\vec\ell_X(\vec\alpha)).
    \end{equation*}
\end{definition}
\begin{remark}
    Note that a fixed term in the sum is only defined in Teichm\"uller space $\cT_g$, but the summation over the orbit $\mcg . \bar\g$ makes it invariant under the mapping class group, and hence is a well-defined function on the moduli space $\cMg$.
\end{remark}
The celebrated \textit{Mirzakhani integration formula} allows to integrate geometric functions with respect to the Weil--Petersson volume form, as an integral over $\R_{\geq 0}^k$.
To write this formula explicitly we need to understand the surface post cutting $\Sigma_g$ by the curves in $\vec\g$. We may, and will, assume that $\vec \g$ consists of its geodesic representatives.

The cut surface $\Sigma_g -\vec\g$ may be written as the disjoint union $\bigsqcup_{i=1}^q \Sigma_{g_i,n_i}$ of its connected components. The $k$ curves of $\vec\g$ form the $2k$ boundary components of $\bigsqcup_{i=1}^q \Sigma_{g_i,n_i}$. If the multicurve $\vec\g$ had lengths $\vec x\in \R_{> 0}^k$ on $X$, then these lengths become the boundary lengths of the (maybe disconnected) surface $\Sigma_g-\bar\g$. Therefore each component has a length vector $\vec x_i \in \R_{> 0}^{n_i}$. We then define
\begin{equation*}
    V_g(\vec \g, \vec x):=\prod_{i=1}^q V_{g_i,n_i}(\vec x_i).
\end{equation*}
\begin{theorem}[\cite{Mirzakhani2006simplegeodesics}]\label{thm: mirzakhani integration formula}
    Let $g\geq3$ and $F^{\vec \g}$ be a geometric function. Then its integral over moduli space with respect to the Weil--Petersson volume equals
    \begin{equation*}
        \int_\cMg F^{\vec \g}(X)\d \vol_g^\WP(X) = \int_0^\infty F(x) V_g(\vec\g,\vec x) x_1\cdots x_k \d x_1 \cdots\td x_k.
    \end{equation*}
\end{theorem}
In the sequel, we will only require computing integrals of multicurves consisting of a single closed geodesic, on surfaces $X$ of genus $g$. A simple closed geodesic $\g$ is characterised by its topological type, and we say it is:
\begin{itemize}
    \item \emph{simple nonseparating} (or $\sns$ for short) if $X - \g$ is connected;
    \item or \emph{simple separating} (or $\ssep$ for short) \emph{of type} $i\in \{1,\ldots, g-1\}$ if $X - \g$ has two connected components: on the left side of $\g$ a surface of signature $(i, 1)$ and on its right a surface of signature $(g-i,1)$.
\end{itemize}
With these definitions we have that
\begin{equation}\label{eq: volume depending on topological types}
    V_g(\g,x) = \begin{cases}
        V_{g-1,2}(x,x) & \text{ if } \g \; \sns \\
        V_{i,1}(x)V_{g-i,1}(x) & \text{ if } \g \;\ssep \text{ of type }i.
    \end{cases}
\end{equation}
\subsection{Volume estimates}
To estimate expectations over moduli space, understanding the asymptotic behaviour of Weil--Petersson volumes is necessary.
We will use the following identities.

For large $g$, one has the following explicit asymptotics (c.f.\ \cite[Proposition 3.1]{MirzakhaniPetri2019} and \cite[Footnote p.\ 3]{AnantharmanLauraCorrectAsymp})
\begin{equation}\label{eq: wp ratio Vgn with lengths}
        \frac{V_{g,n}(x_1,\ldots,x_n)}{V_{g,n}}= \prod_{i=1}^n \frac{\sinh(x_i/2)}{x_i/2} + O_n\prt{\frac{\sum_{i=1}^n x _i}{g}\exp\prt{\frac{\sum_{i=1}^nx_i}{2}}}.
\end{equation}

The following upper bound is implicit in the work of Mirzakhani \cite{Mirzakhani2013}, and it is explicitly stated and proven in Nie--Wu--Xue \cite[Lemma 22]{NieWuXue},
\begin{equation}\label{eq: trivial wp ratio Vgn with lengths}
    \frac{V_{g,n}(x_1,\ldots, x_n)}{V_{g,n}} \leq \prod_{i=1}^n\frac{\sinh(x_i/2)}{x_i/2}.
\end{equation}

For large $g$, one has the following asymptotic (c.f.\ \cite[Theorem 1.4]{mirzakhani_towards_2015}):
    \begin{equation}\label{eq: volume ratio V g-1 n+2 g n}
        \frac{V_{g-1,n+2}}{V_{g,n}}=1+O_n\prt{\frac{1}{g}}
    \end{equation}

Finally, one has the following asymptotic on the Weil--Petersson volumes of surfaces cut out by simple-separating multicurves (c.f.\ \cite[Lemma 3.2]{MirzakhaniPetri2019}).

Fix $q\geq 2$. Then 
\begin{equation}\label{eq: volume ratio separating multicurves}
        \frac{1}{V_g}\sum \prod_{i=1}^q V_{g_i,b_i} =O\prt{ \frac{1}{g^{q-1}}}
\end{equation}
where the sum runs over all topological types of multicurves $\vec \g$ cutting it into $q$ pieces $\Sigma_{g_i,b_i}$ with $b_i\geq 1$,  and $\sum_{i=1}^q 2-2g_i-b_i=2-2g$.
\section{The eigenfunction statistic}\label{sec: eigenfunction stat}
In this section we introduce our setup, and take the first steps in computing the asymptotics of the variance of our local Weyl law. 
\subsection{Pointwise local Weyl law}
Let $f:\C\ra\C$ be an even function such that $\hf\in C_c^\infty([-1,1])$. Then define for $L, \tau >0$ the function
\begin{equation*}
    h_{L,\tau}(r)=h(r):=f(L(r-\tau))+f(L(r+\tau)),
\end{equation*}
with Fourier transform
\begin{equation}\label{eq: hhat form}
    \hh(u)=\frac{2}{L}\cos(\tau u)\hf\prt{\frac{u}{L}}.
\end{equation}

We study the random variable
\begin{equation}\label{eq: definition of N(X,z)}
    N_{L,\tau}(X,z)=N(X,z):= \sum_{j\geq 0}h(r_j)|\phi_j(z)|^2 .
\end{equation}
Note that this quantity is independent of the choice of eigenbasis, as any change of basis is unitary.

By the pre-trace formula, Theorem \ref{thm: pre-trace formula}, the quantity $N(X,z)$ may be split into 
\begin{equation}\label{eq: splitting of N}
    N(X,z)= \bar N + N_{\osc}(X,z)
\end{equation}
with 
\begin{equation*}
\bar N = \frac{1}{4\pi} \int_{-\infty}^\infty h(r)\tanh(\pi r)r\d r
\end{equation*}
being the constant ``topological term" depending only on the parameters $f, L,\tau$, and not $\xz$, and
\begin{equation}\label{eq: def of Nosc}
    N_{\osc}(X,z):= \sum_{\g\in \pi_1\xz-\{\id\}} k(\ell(\g)).
\end{equation} 
These are the ``oscillatory" (or ``geometric") terms.

The kernel $k$ can be written as $k = \mathcal S^{-1}(h)$, which by  (\ref{def: selberg transform}) equals 
\begin{equation}\label{eq: k but expanded with h'}
    k(x)=-\frac{1}{\sqrt 2 \pi}\int_x^\infty \frac{\hh'(y)}{\sqrt{\cosh y-\cosh x}}\td y.
\end{equation}
\begin{lemma}[Properties of $k$]\label{lem: props of k} 
One has that $k\in C^\infty(\R_{\geq 0})$, and that there exists a constant $C = C(f,L,\tau)$ such that for all $x\in \R_{\geq 0}$, 
\begin{equation*}
    |k(x)| \leq C \1_{[0,L]}(x).
\end{equation*}
\end{lemma}
\begin{proof}
    It is a classical fact that $k\in C^\infty(\R_{\geq 0})$ (c.f.\ \cite[Proposition 3]{Marklof2011}).
    It follows immediately from $\supp \hf\subseteq[-1,1]$ that $\supp \hh \subseteq[-L,L]$, by (\ref{eq: hhat form}). Thus $\supp \hh' \subseteq[-L,L]$ which yields by the identity (\ref{eq: k but expanded with h'}) that $\supp k \subseteq [0,L]$. This, together with the smoothness of $k$, yields the claimed bound.
\end{proof}
\subsection{The probability space}\label{subsec: the probability space}
Let $\cMg$ be endowed with the normalised Weil--Petersson measure, denoted by
\begin{equation*}
    \mu_g^\WP:= \frac{1}{V_g}\td \vol_g^\WP.
\end{equation*}

We define our sample space $\Omega_g$ as
\begin{equation*}
    \Omega_g:= \{(X,z):X \text{ is a genus $g$ hyperbolic surface}, z\in X \}/\text{marked isometries}.
\end{equation*}
By marked isometries we mean that $(X,z)\sim (X',z')$ if there exists an orientation-preserving isometry $i:X\to X'$ such that $i(z) =z'$. 

We equip it with the probability measure that is $\mu_g^\WP$ on the first coordinate, and the normalised area measure $\td\mu = \frac{\td z}{\vol X} = \frac{\td z}{4\pi(g-1)}$ on the second.
This turns $\Omega_g$ into a probability space. 

When taking expectations over this probability space, we will always first integrate over the surface, to yield a Weil--Petersson expectation over $\cMg$. Explicitly, if $M:\Omega_g \to \R$ is a random variable, then
\begin{equation}\label{eq: expectation in our model definition}
    \Eg\left[M\xz \right] = \E_g^\WP[\tilde M(X)],
\end{equation}
where $\tilde M$ is the associated \emph{normalised space average}
\begin{align*}
    \tilde M:\cMg &\to \R \\
    X &\mapsto  \int_X M\xz \frac{\td z}{4\pi(g-1)}.
\end{align*}
\subsection{Expectation of local Weyl law} 
Using the pre-trace formula (\ref{eq: splitting of N}), we can compute the expectation
\begin{equation*}
    \Eg[N\xz] = \bar N + \Eg[N_\osc\xz].
\end{equation*}
By Selberg's trace formula, Lemma \ref{lem: selberg trace formula}, and work of Rudnick \cite{rudnick2023goestatisticsmodulispace}, the expectation $\Eg\left[ N_\osc(X,z)\right]$ may be explicitly estimated.

\begin{lemma}\label{lem: rudnick main expectation}
    For fixed $f,\tau>0$ and $L\geq 1$, as  $g\to\infty$,
    \begin{align}\label{eq: expectation with main and off terms}
    \begin{split}
        \Eg \left[N_\osc\xz\right] =
        \frac{1}{4\pi (g-1)}I_f(L,\tau) + O_{L,f,\tau}(g^{-2}),
    \end{split}
    \end{align}
    where 
    \begin{equation*}
        I_f(L,\tau):= \frac{4}{L}\int_0^\infty \sum_{m=1}^\infty \hf\prt{\frac{mx}L} \frac{\sinh^2(x/2)}{\sinh(mx/2)} \cos(\tau mx)\d x.
    \end{equation*}
\end{lemma}
\begin{proof}
    By (\ref{eq: expectation in our model definition}),
    \begin{equation*}
        \Eg[N_\osc\xz] = \frac{1}{4\pi (g-1)}\int_\cMg \int_X N_\osc(X,z)\d z \td\mu_g^\WP(X).
    \end{equation*}
    By definition (\ref{eq: def of Nosc}),
    \begin{equation*}
        \Eg[N_\osc\xz]=\frac{1}{4\pi(g-1)}\int_\cMg \int_X \sum_{\g\in\pi_1\xz- \{ \id\}}k(\ell(\g))\d z \td\mu_g^\WP(X). 
    \end{equation*}
    By (\ref{eq: integral identity stf}) of Lemma \ref{lem: selberg trace formula} this equals,
    \begin{equation*}
         \frac{1}{4\pi(g-1)}\int_\cMg \sum_{\g\in[\cP]}\sum_{n=1}^\infty \frac{\ell(\g)}{2\sinh\prt{\frac{n\ell}{2}}}{\hh(n\ell(\g))} \d\mu_g^\WP(X).
    \end{equation*}
    The integral over moduli space is exactly the term computed in \cite[Prop.\ 4.1]{rudnick2023goestatisticsmodulispace}, whence the statement of the lemma follows.
\end{proof}
\subsection{The variance}
As $N = \bar N + N_\osc$ and $\bar N$ is a constant throughout $\Omega_g$, one has that
\begin{equation*}
    \Varg (N) = \Varg(N_\osc).
\end{equation*}
Thus by (\ref{eq: def of Nosc})
\begin{align}\label{eq: form of variance}
    \begin{split}
        \Varg \prt{N_\osc} = \Eg\left[\sum_{\g_0,\g_1\in\pi_1\xz-\{\id\}} k(\ell(\g_0))k(\ell(\g_1))\right]-\prt{\Eg\left[N_\osc\xz\right]}^2.
    \end{split}
\end{align}
Hence the variance reduces to an expectation over lengths of pairs of geodesic loops, based at the same point.

We say a pair $\g_0,\g_1\in \pi_1\xz-\{\id\}$ of geodesic loops is \emph{diagonal} if both are nonzero integer powers of the same primitive loop, and are \emph{off-diagonal} otherwise. We denote the sum over these types of pairs by $\D\xz$ and $\od\xz$, respectively. Explicitly, these are
\begin{align}\label{eq: def of D and OD}
    \begin{split}
        \D(X,z) &:=\sum_{\g\in \cP}\sum_{n,|m|\geq 1} k(\ell_\xz(\g^n))\;k(\ell_\xz(\g^m))\\
        \od(X,z) &:= \sum_{\gamma_0\neq \gamma_1^{\pm 1}\in \cP}\sum_{n,m\geq 1}k(\ell_\xz(\g^n))\;k(\ell_\xz(\g^m)).
    \end{split}
\end{align}
The summation indices are used to avoid double counting. In the diagonal terms, we first choose an oriented primitive loop and then sum over its positive and negative iterates; in the off-diagonal terms, we sum over powers of two oriented geodesic loops, neither being positive or negative iterates of each other.

Theorem \ref{thm: main theorem on the variance} follows from combining, with Lemma \ref{lem: rudnick main expectation}, the following two propositions about the expectations of the diagonal and off-diagonal terms.
\begin{proposition}\label{prop: main diagonal term claim}
    For fixed $\tau>0, f$, and $L\geq 1$, as $g\to\infty$, the expectation of the diagonal terms satisfies
    \begin{align*}
        \Eg[\D(X,z)] = \frac{1}{4\pi(g-1)} \frac{\tau \tanh \pi \tau}{\pi L} \|f\|_{L^2(\R)}^2 + O_{f,\tau}\prt{\frac{1}{L^2 g}}
        +O_{f,L,\tau}\prt{\frac{1}{g^2}}.
    \end{align*}
\end{proposition}
\begin{proposition}\label{prop: off diagonal term claim}
    For fixed $\tau>0, f$, and $L\geq 1$, as $g\to\infty$, the expectation of the off-diagonal terms is bounded by
    \begin{equation*}
        \Eg[\od(X,z)] = O_{f,L,\tau}\prt{\frac{1}{g^2}}.
    \end{equation*}
\end{proposition}
The proofs of these propositions can be found in \hyperref[proof: linear_proof of main diagonal proposition]{Section \ref{sec: diagterms}} and \hyperref[prf: n proof of main offdiag proposition]{Section \ref{sec: offdiagterms}} respectively.
\begin{proof}[Proof of Theorem \ref{thm: main theorem on the variance}]\label{proof: proof of main theorem on the variance}
    For the expectation of $N\xz$, it follows from (\ref{eq: splitting of N}) and Lemma \ref{lem: rudnick main expectation} that
    \begin{equation*}
        \Eg \left[N\xz\right] = \frac{1}{4\pi}\imii h(r)r\tanh\pi r \d r + O_{f,L,\tau}\prt{\frac 1g}.
    \end{equation*}
    
    For the variance, by (\ref{eq: form of variance}) and the linearity of the expectation
    \begin{equation}\label{eq: linear_variance split}
        \begin{split}
            \Varg(N\xz)&=\Varg (N_\osc\xz)\\& = \Eg[\D \xz] + \Eg[\od\xz] - \prt{\Eg[N_\osc]}^2.
        \end{split}
    \end{equation}
    By Lemma \ref{lem: rudnick main expectation}, the square of the expectation equals
    \begin{equation*}\label{eq: exp nosc square is g -2}
        \prt{\Eg[N_\osc]}^2 = \prt{\frac{1}{4\pi (g-1)}I_f(L,\tau) + O_{L,\tau}(g^{-2})}^2 = O_{f,L,\tau}(g^{-2}).
    \end{equation*}
    Thus after plugging in the expectations of the diagonal and off-diagonal terms, which are  Propositions \ref{prop: main diagonal term claim} and \ref{prop: off diagonal term claim} above, into the variance formula (\ref{eq: linear_variance split}), one arrives at the theorem statement:
        \begin{align*}
            \Varg \prt{N\xz} &= \frac{1}{4\pi(g-1)} \frac{\tau \tanh \pi \tau}{\pi L} \|f\|_{L^2(\R)}^2 + O_{f,\tau}\prt{\frac{1}{Lg}}
            +O_{f,L,\tau}\prt{\frac{1}{g^2}}.
        \end{align*}
\end{proof}

\section{Asymptotics of the diagonal terms}\label{sec: diagterms}
The purpose of this section is to prove Proposition \ref{prop: main diagonal term claim}, the asymptotics of the expectation of the diagonal terms $\Eg \left[\D\xz\right]$. This expectation, crucially, is amenable to exact computations.

The section is split as follows.
\begin{itemize}
    \item[(i)] Section \ref{subsec: linear_special functions}: defines special functions which will arise in our calculations, and collects some necessary results.
    \item[(ii)] Section \ref{subsec: computation of integral on a surface}: computes the form of $\int_X \D\xz \d z$ in terms of a particular function  $H:\R_{>0}\ra\R$ of the lengths of the (oriented) closed geodesics on the surface. Furthermore, we introduce Proposition \ref{prop: asymp of relevant H integral}, which is central, and forms the basis of the asymptotics of Proposition \ref{prop: main diagonal term claim}.
    \item[(iii)] Section \ref{subsec: linear_splitting into topological types}: proves Proposition \ref{prop: main diagonal term claim} using Mirzakhani's work and Proposition \ref{prop: asymp of relevant H integral}.
    \item[(iv)] Section \ref{subsec: asymptotics and stationary phase}: further develops the form of the function $H$ above, and proves necessary asymptotics of auxiliary functions which appear in this form.
    \item[(v)] Section \ref{subsec: asymptotics of main integral}: proves Proposition \ref{prop: asymp of relevant H integral}.
\end{itemize}
\subsection{Special functions}\label{subsec: linear_special functions}
The following function will be a key character in the sequel.
\begin{definition}[Complete elliptic integral of the first kind]\label{def: ellipticK}
    The complete elliptic integral of the first kind is defined as
    \begin{equation}\label{def:ellipticK}
        K(x):=\int_0^{\frac{\pi}{2}}\frac{\td\theta}{\sqrt{1-x^2\sin^2\theta}}
    \end{equation}
    for $x\in[0,1)$.
\end{definition}
The elliptic $K$-function has the following properties, which can be found in e.g.\ \cite{NIST:DLMF}.
\begin{lemma}[Properties of $K$]\label{lem: props of K}\;
\begin{enumerate}
    \item One has the Taylor expansion 
    \begin{equation}\label{eq: ellipticK taylor expansion}
        K(x)=\frac{\pi}{2}\sum_{n=0}^\infty \frac{\prt{\frac 12}_n\prt{\frac 12}_n}{(n!)^2} x^{2n},
    \end{equation}
    where $(z)_n = z(z+1)\cdots (z+n-1)$ is the Pochhammer symbol.
    \item One has
    \begin{equation*}
        K(x) \leq \frac \pi 2 - \frac 12 \log \prt{1-x^2}.
    \end{equation*}
    with the divergence as $x\uparrow 1$ of $K(x)$ being asymptotic to
    \[K(x)\underset{x\uparrow1}\sim -\frac12\log(1-x^2).\]
\end{enumerate}
\end{lemma}
We now move on to define functions related to the hyperbolic distance which will be of use.

The hyperbolic distance between two points has several different expressions. We will use the following simple one:
\begin{definition}\label{def: w and W}
    For $z=x+iy\in \H$ and $l> 0$ we set
    \begin{equation*}
        w_l(x/y):= d_\H (z,e^l\dot z), \text{ and } W_l(x/y) \text{ its inverse}. 
    \end{equation*}
\end{definition}
Explicitly,
\begin{equation}\label{eq: linear_def of w ell}
    w_l(s):= \ach \prt{\cosh l + (\cosh l-1)s^2} ,\text{ for } s\in \R,
\end{equation}
and its inverse may be expressed as
\begin{equation}\label{eq: linear_def of W ell}
    W_l(s):= \sqrt{\frac{\cosh s-\cosh l}{\cosh l -1}}, \text{ for } |s|\geq l.
\end{equation}

Finally, for reasons that will become apparent in Section \ref{subsec: asymptotics of main integral}, we prove the following (satisfying) integral identity with respect to the elliptic $K$-function.
\begin{lemma}\label{lem: linear_cosine K integral}
    For all $s\in \R$, the following identity holds
    \begin{equation*}
        \izi \cos(s v) \izi \frac{K\prt{\sqrt{\frac{e^y-1}{e^{y+v}-1}}}}{\sqrt{e^{y+v }-1}} \d y\td v=\begin{cases}
            \frac{\pi^2}{4s}\tanh\pi s & s \neq 0 \\
            \frac{\pi^3}{4} & s = 0.
        \end{cases}
    \end{equation*}
\end{lemma}
\begin{proof}
    Let $s \neq 0$. We first compute the integral
    \begin{equation*}
        \izi \frac{K\prt{\sqrt{\frac{e^y-1}{e^{y+v}-1}}}}{\sqrt{e^{y+v}-1}}\d y.
    \end{equation*}
    The change of variables $x = \sqrt{\frac{e^y-1}{e^{y+v}-1}}$, yields that $x$ goes from $0$ to $e^{-v/2}$, the Jacobian is $\frac{2x(e^v -1)}{(1-x^2)(1-x^2e^v)}$, and the square root term $\frac{1}{\sqrt{e^{y+v}-1} }$ equals $\sqrt{\frac{1-x^2e^v}{e^v-1}}$. Plug them in to arrive at
    \begin{equation*}
        \int_0^{e^{-v/2}}K(x) \frac{\sqrt{1-x^2e^v}}{\sqrt{e^v-1}}\frac{2x(e^v-1)}{(1-x^2)(1-e^vx^2)} \d x = 2\sqrt{1-e^{-v}}\int_0^{e^{-v/2}}\frac{xK(x)}{(1-x^2)\sqrt{e^{-v}-x^2}}\d x.
    \end{equation*}
    Using \cite[6.153]{gradshteyn2007} yields an exact form for the integral, namely
    \begin{equation*}
        2\sqrt{1-e^{-v}} \cdot \frac{\pi}{4} \frac{1}{\sqrt{1-e^{-v}}} \log \prt{\frac{1+e^{-v/2}}{1-e^{-v/2}}}.
    \end{equation*}
    
    Thus the double integral equals
    \begin{equation}\label{eq: cos log seemek}
        \frac{\pi}{2}\izi \cos(s v)  \log \prt{\frac{1+e^{-v/2}}{1-e^{-v/2}}}\d v.
    \end{equation}
    Applying the formulae for the cosine transforms of $\log \prt{1+e^{-v/2}}$ and $\log\prt{1-e^{-v/2}}$ \cite[1.4.43, 1.4.44]{oberhettinger1990tables} yields that the above equals
    \begin{equation*}
        \frac{\pi}{2}\prt{\prt{\frac{1}{4s^2}-\frac{\pi}{2s}\csch\prt{2\pi s}}- \prt{\frac{1}{4s^2}-\frac{\pi}{2s}\coth \prt{2\pi s}}},
    \end{equation*}
    which by hyperbolic function identities equals the identity of the lemma.

    We extend the definition to $s = 0$ by continuity, and indeed, the double integral converges at $s=0$ (as can be seen in e.g.\ (\ref{eq: cos log seemek})).
\end{proof}
\subsection{Space average of diagonal terms on a deterministic surface}\label{subsec: computation of integral on a surface}
In this section, we prove that the space average of the diagonal terms $\D\xz$ may be written as a weighted sum over the lengths $\ell(\g)$ of primitive closed geodesics $\g$, weighted by a function $H$. This removes all dependency on marked points $z$ on the surface. This is Proposition \ref{prop: diagonal on deterministic}.
We proceed thereafter to prove some basic results for the function $H$, and introduce a certain integral whose evaluation forms a key ingredient in proving Proposition \ref{prop: main diagonal term claim}.

We begin by proving a general formula for space averages of functions dependent on lengths of geodesic loops, freely homotopic to some fixed closed geodesic. This is similar to computations made in the proof of the Selberg trace formula (c.f.\ \cite{Marklof2011}). 
\begin{lemma}\label{lem: general space average}
    Let $X$ be a hyperbolic surface, and $\g$ be a primitive closed geodesic of length $\ell>0$. Let $t\geq 1$ be an integer. Then let $ n_1,\ldots,n_t \geq 1$ be integers, and $q_1,\ldots, q_t:\R_{>0}\ra \R$ be compactly supported functions. Define 
    \begin{equation*}
        M\xz = \sum_{\substack{\alpha \in \pi_1\xz\\ \alpha \simf \g}}\prod_{j=1}^tq_j(\ell_\xz(\alpha^{n_j})).
    \end{equation*}
    Then the space average of $T\xz$ has the following integral form
    \begin{equation*}
        \int_X M\xz \d z  = 2\ell \izi \prod_{j=1}^t q_j\prt{w_{n_j\ell}\prt{ x }} \td x.
    \end{equation*}
\end{lemma}
\begin{remark}
    In the sequel, we will only use the case $t = 2, q_1 = q_2 = k$ of the lemma.
\end{remark}
\begin{proof}
    Fix a fundamental domain $F\subseteq \H$ of $X$, and a corresponding $\G\leq \psltr$. Every point $z\in X$ has a corresponding lift to a point $\tilde z\in F$. 
    
    The closed geodesic $\g$ naturally corresponds to a conjugacy class $[\g] \subseteq \G$. Each geodesic loop $\alpha \in \pi_1\xz$ such that $\alpha \simf \g$ naturally corresponds to an element $\tilde \alpha \in [\g]$ in this conjugacy class, with $\ell_\xz(\alpha) = d(z,\tilde \alpha .z)$.

    Thus
    \begin{equation*}
        I:=\int_X M\xz \d z= \int_F \sum_{\alpha\in [\tilde \g]}\prod_{j=1}^t q_j(d(z,\tilde \alpha^{n_j}.z))\d z.
    \end{equation*}
    The sum is independent of the basepoint $z$, and thus
    \begin{equation*}
        I=\sum_{\alpha\in [\tilde \g]}\int_F \prod_{j=1}^t q_j(d(z,\tilde \alpha^{n_j}.z))\d z.
    \end{equation*}
    By the orbit-stabilizer theorem, with $\G_{\tilde \g} := \{\delta \in \G: \delta^{-1}\tilde \g \delta = \tilde \g \}$ the associated \emph{centralizer}, this equals
    \begin{equation*}
        I=\sum_{\delta\in \G_{\tilde \g}\ba \G}\int_F \prod_{j=1}^t q_j(d(z,\delta^{-1}\tilde \g^{n_j}\delta.z))\d z.
    \end{equation*}
    Changing variables, as $\G$ acts by isometries, this yields that
    \begin{equation*}
       I= \sum_{\delta\in \G_{\tilde \g}\ba \G}\int_{\delta.F} \prod_{j=1}^t q_j(d(z,\tilde \g^{n_j}.z))\d z.
    \end{equation*}
    By definition of a fundamental domain,
    \begin{equation*}
        I=\int_{F_{\tilde \g}} \prod_{j=1}^t q_j(d(z,\tilde \g^{n_j}.z))\d z,
    \end{equation*}
    where $F_{\tilde \g}$ is a fundamental domain for the action of $\G_{\tilde \g}$.

    The fundamental domain $F_{\tilde \g}$ is a hyperbolic cylinder of collar width $\ell = \ell(\g)$ around the axis of $\tilde \g$. One may change coordinates of this integral to be centered around the axis of $\tilde \g$, such that $\tilde \g.z = e^\ell z$. This explicitly yields, with $z = x+iy$, that
    \begin{equation*}
        I=  \int_1^{e^\ell} \imii\prod_{j=1}^t q_j(d(z,e^{n_j\ell}z))\frac{\td x \td y}{y^2}.
    \end{equation*}
    By Definition \ref{def: w and W}
    \begin{equation*}
        I=\int_1^{e^\ell}\imii  \prod_{j=1}^t q_j\prt{w_{n_j\ell}\prt{ x /y}}\frac{\td x \td y}{y^2}.
    \end{equation*}
    Changing variables $x\mapsto yx$, 
    \begin{equation*}
        I=\int_1^{e^\ell} \frac{\td y}{y} \imii \prod_{j=1}^t q_j\prt{w_{n_j\ell}\prt{ x }} \td x.
    \end{equation*}
    By the evenness of $w_{n_j\ell}$, 
    \begin{equation*}
        I = 2\ell \izi \prod_{j=1}^t q_j\prt{w_{n_j\ell}\prt{ x }} \td x,
    \end{equation*}
    which is the statement of the lemma.
\end{proof}
In the spirit of the above result, we define the following weight function.
\begin{definition}[Definition of $H$]\label{def: H}
    For $\ell>0$, define
    \begin{equation}\label{eq: def of H ell in terms of H nm}
        H(\ell) := \sum_{n,m\geq 1} H_{n,m}(\ell)
    \end{equation}
    where
    \begin{equation}\label{eq: def of Hnmell k}
        H_{n,m}(\ell):= 4\ell \izi k(w_{n\ell}(x))k(w_{m\ell}(x))\d x.
    \end{equation}
\end{definition}
We defer proving further properties of $H(\ell)$ to after the following Propostion \ref{prop: diagonal on deterministic}, the main result of this section.  It implies, in other words, that the space average of the diagonal terms is given explicitly as a sum over lengths of primitive closed geodesics on the surface.
\begin{proposition}\label{prop: diagonal on deterministic}
    Let $X$ be a hyperbolic surface. The following identity holds:
    \begin{equation*}
        \int_X \D\xz\d z = \sum_{\substack{\g \text{ prim.}\\ \text{closed geod.}}} H(\ell(\g)).
    \end{equation*}
\end{proposition}
\begin{proof}
    By the definition of $\D\xz$, which is (\ref{eq: def of D and OD}), we have that
    \begin{equation*}
    I:=\int_X \D\xz \d z =\int_X\sum_{\g\in\cP}\sum_{n,|m|\geq 1}k(\ell_\xz(\g^n))k(\ell_\xz(\g^m))\d z.
    \end{equation*}
Organising the elements in $\cP$ by conjugacy classes, i.e.\ by closed geodesics, the above equals
\begin{equation*}
    2\sum_{\substack{\g \text{ prim.}\\ \text{closed geod.}}}\sum_{n,m\geq 1} \int_X \sum_{\substack{\alpha \in \pi_1\xz\\ \alpha \simf \g}} k(\ell_\xz(\alpha^n))k(\ell_\xz(\alpha^m))\d z,
\end{equation*}
where we picked up a factor of $2$ by forgetting the orientation of $\alpha^m$.

By Lemma \ref{lem: general space average} the above equals
\begin{equation*}
    \sum_{\substack{\g \text{ prim.}\\ \text{closed geod.}}}\sum_{n,m\geq 1} H_{n,m}(\ell(\g)),
\end{equation*}
which by Definition \ref{def: H} yields the lemma.
\end{proof}
The function $H$ above has the following simple properties:
\begin{lemma}\label{lem: linear_props of H}
    The function $H$ is continuous on $(0,\infty)$, with $\supp H \subseteq (0,L]$; furthermore, there exists a constant $C=C(f,L,\tau)$ such that
    \begin{equation*}
            H(\ell) \sim \frac{C}{\ell}, \text{ as } \ell\downarrow 0.
    \end{equation*}
\end{lemma}
\begin{proof}
    Let $\ell>0$. Recall by Definition \ref{def: H} that
    \begin{equation*}
        H(\ell) := 4\ell\sum_{n,m\geq 1} \izi k(w_{n\ell}(x))k(w_{m\ell}(x))\d x.
    \end{equation*}
    For all $x\geq 0$ one has that
    \begin{equation}\label{eq: w nell and ell inequality}
        w_{n\ell}(x)\geq n\ell \geq \ell.
    \end{equation}
    This implies, as $\supp k \subseteq [0,L]$ by Lemma \ref{lem: props of k}, that
    \begin{equation}\label{eq: finite sum H}
        H(\ell) = 4\ell \sum_{1\leq n,m\leq \lceil L/\ell \rceil} \izi k(w_{n\ell}(x))k(w_{m\ell}(x))\d x.
    \end{equation}
    This is a finite sum, from which the continuity and support statement follows.
    
    To prove the asymptotic as $\ell\downarrow0$ of $H(\ell)$, we first prove that the following triple integral
    \begin{equation*}
        I_1:= \izi \prt{\izi k(w_u(x))\d u}^2\d x.
    \end{equation*}
    converges.
    To see this, by Lemma \ref{lem: props of k},
    \begin{equation*}
        I_1 \ll_{f,L,\tau}I_2:= \izi \prt{\izi\1_{[0,L]}(w_u(x))\d u}^2 \d x.
    \end{equation*}
    By definition of $w$, Definition \ref{def: w and W}, the inequality $w_u(x) \leq L$ holds if and only if
    \begin{equation*}
        u\leq \ach \prt{1+ \frac{\cosh L -1}{x^2 + 1}}.
    \end{equation*}
    Thus
    \begin{equation*}
        I_2 = \izi\prt{\ach \prt{1+ \frac{\cosh L - 1}{x^2 + 1}}}^2 \td x. 
    \end{equation*}
    Since $\ach (1+x) \leq \sqrt{2x}$,
    \begin{equation*}
        I_2\leq \izi \prt{\sqrt{2 \frac{\cosh L -1}{x^2+1}}}^2 \td x <\infty.
    \end{equation*}

    Returning to $H$, we have
    \begin{equation*}
        H(\ell) = 4\ell \sum_{n,m\geq 1}\izi k(w_{n\ell}(x))k(w_{m\ell}(x))\d x=\frac{4}{\ell}\izi \sum_{n=1}^\infty \ell k(w_{n\ell}(x))\sum_{m=1}^\infty \ell k(w_{m\ell}(x))\d x.
    \end{equation*}
    The sums in the integrand are Riemann sums, and as $\ell \downarrow 0$ they converge to
    \begin{equation*}
        \sum_{n=1}^\infty \ell k(w_{n\ell}(x)) \xrightarrow[\ell\downarrow 0]{}\izi k(w_u(x))\d u
    \end{equation*}
    uniformly, by the compact support of $k$. Thus
    \begin{equation*}
        H(\ell) \underset{\ell\downarrow 0}{\sim} \frac{4I_1}{\ell}.
    \end{equation*}
\end{proof}
An important corollary of the above Lemma \ref{lem: linear_props of H} is the following convergence result.
\begin{corollary}\label{cor: H sinh square integral converges}
    The integral
    \begin{equation*}
        \izi H(\ell)\frac{\sinh^2\prt{\frac\ell2}}{\ell}\d\ell 
    \end{equation*}
    converges absolutely.
\end{corollary}
In fact, the above integral and its evaluation are central to proving Proposition \ref{prop: main diagonal term claim}, and hence to Theorem \ref{thm: main theorem on the variance}. As we shall see in the sequel, the following proposition establishes the main term of both of the aforementioned results.
\begin{proposition}\label{prop: asymp of relevant H integral}
    For fixed $\tau>0, f$, and any $L\geq 1$,
    \begin{equation*}
        \izi H(\ell) \frac{\sinh^2\prt{\frac \ell2}}{\ell}\d\ell = \frac{\tau \tanh\pi \tau}{4\pi L}\|f\|_{L^2(\R)}^2 + O_{f,\tau}\prt{\frac{1}{L^2}}.
    \end{equation*}
\end{proposition}
The proof of this proposition is rather analytically involved, and is the subject of Section \ref{subsec: asymptotics of main integral}.

\subsection{Expectation of diagonal terms}\label{subsec: linear_splitting into topological types}
The aim of this section is to prove Proposition \ref{prop: main diagonal term claim}, given Proposition \ref{prop: asymp of relevant H integral}.

Let us split the (unnormalised) space average of the diagonal terms by topological types:
\begin{align}\label{eq: int of diagonal as sum over geodesic types}
    \begin{split}
        \int_X \D\xz \d z &=\sum_{\substack{\g \text{ prim.}\\ \text{closed geod.}}} H(\ell(\g))\\
        &=S_{\sns} + S_{\ssep} + S_{\ns}.
    \end{split}
\end{align}
The terms $S_{\sns},S_{\ssep},S_{\ns}$ run over the primitive closed geodesics which are simple nonseparating, simple separating, and non-simple respectively (c.f.\ Section \ref{subsec: mirz int formula} for definitions), e.g.\ 
\begin{equation*}
    S_\sns := \sum_{\substack{\g\;\sns}}H(\ell(\g)).
\end{equation*}
For brevity, we also define the following sum 
\begin{equation}\label{eq: def of S_rest}
    S_\rest:= S_\ssep +S_\ns.
\end{equation}

We proceed to prove Proposition \ref{prop: main diagonal term claim}, given the following two lemmas, whose proofs will occupy the rest of this section.
\begin{lemma}[Asymptotics of simple non-separating]\label{lem: linear_expectation of sns}
    For fixed $\tau>0, f$ and $L\geq 1$, as $g\to\infty$,
    \begin{equation*}
        \E_g^\WP[S_\sns] = \frac{\tau \tanh \pi \tau}{\pi L} \|f\|_{L^2(\R)}^2 + O_{f,\tau}\prt{\frac{1}{L^2}}+ O_{f, L,\tau}\prt{\frac 1g}.
    \end{equation*}
\end{lemma}
\begin{lemma}[Asymptotics of rest]\label{lem: linear_expectation of rest}
    For fixed $\tau>0, f$ and $L\geq 1$, as $g\to\infty$,
    \begin{equation*}
        \E_g^\WP \left[ S_\rest \right] = O_{f, L,\tau}\prt{\frac1g}.
    \end{equation*}
\end{lemma}
\begin{proof}[Proof of Proposition \ref{prop: main diagonal term claim}]\label{proof: linear_proof of main diagonal proposition}
    Recall that by definition
    \begin{equation*}
        \Eg[\D\xz] = \frac{1}{4\pi (g-1)}\E_g^\WP \left[\int_X \D\xz \d z \right].
    \end{equation*}
    By (\ref{eq: int of diagonal as sum over geodesic types}) above,
    \begin{equation*}
         \Eg[\D\xz] = \frac{1}{4\pi(g-1)}\prt{\E_g^\WP \left[S_\sns\right]+ \E_g^\WP\left[ S_\rest\right]}.
    \end{equation*}
    Plugging in both Lemmas \ref{lem: linear_expectation of sns} and Lemma \ref{lem: linear_expectation of rest}, yields
    \begin{equation*}
        \Eg[\D\xz] = \frac{1}{4\pi(g-1)}\prt{ \frac{\tau \tanh \pi \tau}{\pi L} \|f\|_{L^2(\R)}^2 + O_{f,\tau}(L^{-2})} + O_{f,L,\tau}\prt{\frac{1}{g^2}},
    \end{equation*}
    which is the statement of the proposition.
\end{proof}
The rest of this section is devoted to proving Lemmas \ref{lem: linear_expectation of sns} and \ref{lem: linear_expectation of rest}.
\begin{proof}[Proof of Lemma \ref{lem: linear_expectation of sns}]
    By definition
    \begin{equation*}
        \E_g^\WP \left[S_\sns\right]=\E_g^\WP \left[\sum_{\substack{\g\;\sns}}H(\ell(\g))\right].
    \end{equation*}
    By Mirzakhani's integration formula, Theorem \ref{thm: mirzakhani integration formula}, for $g>2$, we have that
    \begin{equation*}
        \E_g^\WP \left[S_\sns\right] = \izi H(\ell) \frac{V_{g-1,2}(\ell,\ell)}{V_g}\ell \d \ell,
    \end{equation*}
    since $V_g(\g;\ell) = V_{g-1,2}(\ell,\ell)$ by (\ref{eq: volume depending on topological types}).
    
    By the volume estimates (\ref{eq: wp ratio Vgn with lengths}) and (\ref{eq: volume ratio V g-1 n+2 g n}) 
    \begin{equation*}
        \frac{V_{g-1,2}(\ell,\ell)}{V_g} = \frac{V_{g-1,2}(\ell,\ell)}{V_{g-1,2}}\frac{V_{g-1,2}}{V_g} = \prt{\frac{\sinh^2\prt{\frac{\ell}{2}}}{\prt{\frac{\ell}2}^2}+ O\prt{\frac{\ell}{g}e^{\ell}}}\prt{1+O\prt{\frac 1g}}
    \end{equation*}
    we have
    \begin{equation*}
        \E_g^\WP[S_\sns]  = 4\izi H(\ell) \frac{\sinh^2 \prt{\frac \ell 2}}{\ell}\d \ell+ O\prt{\frac1{g}\izi H(\ell)\ell e^\ell\d\ell}.
    \end{equation*}
    By Lemma \ref{lem: linear_props of H},
    \begin{equation*}
        \E_g^\WP[S_\sns] = 4\izi H(\ell) \frac{\sinh^2 \prt{\frac\ell2}}{\ell}\d\ell + O_{L,f,\tau}\prt{\frac 1 g}.
    \end{equation*}
    The lemma follows by Proposition \ref{prop: asymp of relevant H integral}.
\end{proof}
\begin{proof}[Proof of Lemma \ref{lem: linear_expectation of rest}]
    The lemma follows from combining the following two claims, as by definition $S_\rest = S_\ssep + S_\ns$. Recall that $\tau>0, f$ and $L\geq 1$ are fixed.
    \begin{claim}\label{claim: E[S simple separating] bound}
        $$ \E_g^\WP \left[S_\ssep \right] = O_{f,L,\tau} \prt{\frac{1}{g}}.$$
    \end{claim}
    \begin{claim}\label{claim: E[S nonsimple] bound}
        $$\E_g^\WP \left[S_\ns\right] = O_{f,L,\tau}\prt{\frac{1}{g}}.$$
    \end{claim}
    \begin{proofblack}[Proof of Claim \ref{claim: E[S simple separating] bound}]
        In the simple and separating case, there are exactly $g-1$ topological types, determined by the genus of the component lying to the left of the closed geodesic.
        
        Using Mirzakhani's integration formula yields
        \begin{equation*}
            \E_g^\WP \left[ S_{\ssep}\right]= \sum_{i=1}^{g-1}\frac{1}{V_g}\int_0^\infty H(\ell)V_{i,1}(\ell)V_{g-i,1}(\ell) \ell \d\ell.
        \end{equation*}
        One has
        \begin{equation*}
            \frac{V_{i,1}(\ell)V_{g-i,1}(\ell)}{V_g}=\frac{V_{i,1}(\ell)}{V_{i,1}}\frac{V_{g-i,1}(\ell)}{V_{g-i,1}} \frac{V_{i,1}V_{g-i,1}}{V_g}.
        \end{equation*}
        By (\ref{eq: trivial wp ratio Vgn with lengths}) 
        
        \begin{equation*}
            \frac{V_{i,1}(\ell)}{V_{i,1}} \leq \frac{\sinh \prt{ \ell/2}}{\ell /2},
        \end{equation*}
        and by (\ref{eq: volume ratio separating multicurves}),
        \begin{equation*}
            \sum_{i=1}^{g-1} \frac{V_{i,1} V_{g-i,1}}{V_g} \ll \frac 1g,
        \end{equation*}
        we have
        \begin{align*}
            \E_g^\WP \left[ S_{\ssep}\right] & \ll\frac 1g\int_0^\infty |H(\ell)|\frac{\sinh^2 \frac{\ell}{2}}{\ell} \d\ell.
        \end{align*}
        By Corollary \ref{cor: H sinh square integral converges}, the claim follows.
    \end{proofblack}
    \begin{proofblack}[Proof of Claim \ref{claim: E[S nonsimple] bound}]
        Any non-simple closed geodesic has length $\ell \geq 4\ash(1)=: c$, c.f.\ \cite[Ch.\ 4 §2]{buser1992geometry}. Thus, by Lemma \ref{lem: linear_props of H}
        \begin{align*}
            \E_g^\WP\left[\sum_{\g \; \ns} H(\ell(\g))\right] &\ll_{f,L,\tau,c} \E_g^\WP\left[\sum_{\g \; \ns}\1_{[0,L]}(\ell(\g))\right],
        \end{align*}
        which Mirzakhani and Petri proved in \cite[Proposition 4.5]{MirzakhaniPetri2019} is $\ll_L 1/g$.
    \end{proofblack}
\end{proof} 
\setcounter{claim}{0}
\subsection{Asymptotics of auxiliary functions}\label{subsec: asymptotics and stationary phase}
To prove Proposition \ref{prop: asymp of relevant H integral}, we require developing the form of the function $H$ further, and proving asymptotics of related expressions. To this end, we define the following functions.
\begin{definition}\label{def: A and Q}
    For integers $n,m\geq 1$ and $\ell,y,v>0$, define
    \begin{equation*}
        Q_{n,m}(\ell,y,v) := \frac{W_{n\ell}(y+n\ell)}{W_{m\ell}(y+v+m\ell)}.
    \end{equation*}
    Now, for integers $n,m\geq 1$ and $\ell,y,v>0$ such that $Q_{n,m}(\ell,y,v) \neq 1$, define
    \begin{align*}
    A_{n,m}(\ell,y,v)&: = \1_{Q_{n,m}(\ell,y,v)< 1}\cdot  \frac{K\prt{Q_{n,m}(\ell,y,v})}{\sqrt{\cosh (y+v+m\ell)-\cosh m\ell}\sqrt{\cosh n\ell -1}} \\
        & + \1_{Q_{n,m}(\ell,y,v)>1}\cdot \frac{K\prt{\frac{1}{Q_{n,m}(\ell,y,v)}}}{\sqrt{\cosh (y+n\ell)-\cosh n\ell}\sqrt{\cosh m\ell -1}},
    \end{align*}
where $K$ is the elliptic integral of the first kind, Definition \ref{def: ellipticK}, and $W$ is defined in Definition \ref{def: w and W} by (\ref{eq: linear_def of W ell}).
\end{definition}
The relevance of these functions becomes apparent in the following lemma.
\begin{lemma}\label{lem: form of Hnm in terms of Anm}
    The function $H$ is given explicitly by the following expression:
    \begin{equation*}
        H(\ell) = \frac{4\ell}{\pi^2}\sum_{n,m\geq 1}\izi\izi \hh'(y+n\ell)\hh'(y+v+m\ell) A_{n,m}(\ell,y,v)\d y\td v.
    \end{equation*}
\end{lemma}
\begin{proof}
First recall by (\ref{eq: k but expanded with h'})
\begin{equation*}
    k(x) = -\frac{1}{\sqrt 2\pi} \int_x^\infty \frac{\hh'(y)}{\sqrt{\cosh y -\cosh x}}\td y.
\end{equation*}
Plugging this into the definition of $H_{n,m}$,
\begin{align*}
        H_{n,m}(\ell)&= 4\ell \izi k(w_{n\ell}(x))k(w_{m\ell}(x))\d x\\
        &=\frac{2\ell}{\pi^2} \int_0^\infty \int_{w_{n\ell}(x)}^\infty \int_{w_{m\ell}(x)}^\infty \frac{\hh'(y_1)\hh'(y_2)}{\sqrt{\cosh y_1-\cosh w_{n\ell}(x)}\sqrt{\cosh y_2-\cosh w_{m\ell}(x)}}\d y_2\td y_1\td x.
\end{align*}

The region of integration for this integral is 
$$\{(x,y_1,y_2):x>0, y_1>w_{n\ell}(x), y_2>w_{m\ell}(x) \}.$$ 
By inverting $w$, this region is the same as the region
\begin{equation*}
    \{(x,y_1,y_2):y_1>n\ell, y_2>m\ell, x<\min \{W_{n\ell}(y_1),W_{m\ell}(y_2) \} \}.
\end{equation*}
Thus by Fubini 
\begin{align*}
        H_{n,m}(\ell)=\frac{2\ell}{\pi^2}&\int_{n\ell}^\infty\int_{m\ell}^\infty \hh'(y_1)\hh'(y_2)\\ &\int_0^{\min\{W_{n\ell}(y_1),W_{m\ell}(y_2)\}} \frac{\td x\td y_2\td y_1}{\sqrt{\cosh y_1-\cosh w_{n\ell}(x)}\sqrt{\cosh y_2 -\cosh w_{m\ell}(x)}},
\end{align*}
which, by the definition of $w$ in (\ref{eq: linear_def of w ell}), 
\begin{align*}
        H_{n,m}(\ell)&=\frac{2\ell}{\pi^2}\int_{n\ell}^\infty\int_{m\ell}^\infty \hh'(y_1)\hh'(y_2)\int_0^{\min\{W_{n\ell}(y_1),W_{m\ell}(y_2)\}}\\ & \frac{\td x\td y_2\td y_1}{\sqrt{\cosh y_1-\cosh n\ell -(\cosh n\ell -1)x^2}\sqrt{\cosh y_2-\cosh m\ell -(\cosh m\ell -1)x^2}}.
\end{align*}

Serendipitously, direct evaluation of the innermost integral yields a closed form expression in terms of the elliptic integral of the first kind $K$, namely:
\begin{claim*}
    If $a,b,c,d>0$ with $\frac{a}{c}\neq \frac{b}{d}$ then 
    \begin{align*}
        \begin{split}
            \int_0^{\min \{\frac{a}{c}, \frac{b}{d} \}}\frac{\td x}{\sqrt{a^2-c^2x^2}\sqrt{b^2-d^2x^2}}=\1_{\frac ac < \frac bd}\frac{K\prt{\frac{ad}{bc}}}{bc} + \1_{\frac{a}{c}> \frac{b}{d}}\frac{K\prt{\frac{bc}{ad}}}{ad}.
        \end{split}
    \end{align*}
\end{claim*}
Thus, setting 
    \begin{align*}
        a &= \sqrt{\cosh y_1-\cosh n\ell}, &b &= \sqrt{\cosh y_2-\cosh m\ell},\\
        c &= \sqrt{\cosh n\ell -1}, &d&=\sqrt{\cosh m\ell -1},
    \end{align*}
yields that 
\begin{align*}
        H_{n,m}(\ell)=\frac{2\ell}{\pi^2}\int_{n\ell}^\infty\int_{m\ell}^\infty\hh'(y_1)\hh'(y_2) \bar A_{n,m}(\ell,y_1,y_2)\td y_2\td y_1,
\end{align*}
where
\begin{align*}
    \bar A_{n,m}(\ell,y_1,y_2):=&\1_{W_{n\ell}(y_1)<W_{m\ell}(y_2)}\cdot  \frac{K\prt{\frac{W_{n\ell}(y_1)}{W_{m\ell}(y_2)}}}{\sqrt{\cosh y_2-\cosh m\ell}\sqrt{\cosh n\ell -1}} \\
        & + \1_{W_{n\ell}(y_1)>W_{m\ell}(y_2)}\cdot \frac{K\prt{\frac{W_{m\ell}(y_2)}{W_{n\ell}(y_1)}}}{\sqrt{\cosh y_1-\cosh n\ell}\sqrt{\cosh m\ell -1}}.
\end{align*}

Changing variables $y_1\mapsto y_1-n\ell$ and $y_2 \mapsto y_2-m\ell$ to arrive at
\begin{align*}
    H_{n,m}(\ell)=\frac{2\ell}{\pi^2}\izi \izi \hh'(y_1+n\ell)\hh'(y_2+m\ell)\bar A_{n,m}(\ell,y_1+n\ell,y_2+m\ell)\d y_2 \td y_1.
\end{align*}

Recall that we are interested in the form of $H(\ell) = \sum_{n,m\geq 1} H_{n,m}(\ell)$, which by the above equals
\begin{align*}
    H(\ell) &=  \frac{2\ell}{\pi^2}\sum_{n,m\geq 1}\izi \izi \hh'(y_1+n\ell)\hh'(y_2+m\ell)\bar A_{n,m}(\ell,y_1+n\ell,y_2+m\ell)\d y_2 \td y_1.
\end{align*}
By the symmetry in the double summation $n\leftrightarrow m$ and $\bar A_{n,m}$, we may assume $y_1<y_2$ in the integral, and gain a factor of two.
Hence, if we change variables $y_2 \mapsto y_1+v, y_1 \mapsto y$, then by Definition \ref{def: A and Q} of $A_{n,m}$,
\begin{equation*}
    H(\ell)=\frac{4\ell}{\pi^2}\sum_{n,m\geq 1}\izi \izi \hh'(y+n\ell)\hh'(y+v+m\ell) A_{n,m}(\ell,y,v)\d v \td y,
\end{equation*}
which, after switching the order of integration by absolute convergence (c.f.\ Corollary \ref{cor: H sinh square integral converges}), is the statement of the lemma.
\end{proof}
By the above, $H$ is defined by two auxiliary functions: the functions $A_{n,m}$ and $\hh'(\cdot) \cdot \hh'(\cdot)$. The rest of this section is devoted to proving the relevant $\ell$-asymptotics of these, and it is split as follows:
\begin{itemize}
    \item[(i)] Section \ref{subsubsec: asymptotics and stationary phase}: proves asymptotics of $A_{n,m}(\ell,y,v)$ in the $\ell$-variable.
    \item[(ii)] Section \ref{subsubsec: development of hh prime squared}: develops the form of $\hh'(\cdot) \cdot \hh'(\cdot)$ to be used in a forthcoming stationary phase argument.
\end{itemize} 
\subsubsection{Properties of $A_{n,m}$}\label{subsubsec: asymptotics and stationary phase}
To understand the asymptotics of $A_{n,m}$, we first prove the asymptotics of $Q_{n,m}$ which appears in its definition, Definition \ref{def: A and Q}.
\begin{lemma}[Asymptotics of $Q_{n,m}$]\label{lem: asymp of Q}
    Let $\eps>0$. Uniformly in $y,v>0$ and $n,m\geq 1$, for any $\ell\geq \eps$
    \begin{align*}
        Q_{n,m}(\ell,y,v) &= \sqrt{\frac{e^y-1}{e^{y+v}-1}}\prt{1+ O_\eps(e^{-\ell\min \{ n,m\}})}.
    \end{align*}
    Furthermore, as $\ell\downarrow 0 $, $Q_{n,m}(\ell,y,v)$ has the following leading order behaviour,
    \begin{align*}
        Q_{n,m}(\ell,y,v) &\underset{\ell \downarrow 0}{\sim} \frac{m}{n} \frac{\sinh\prt {\frac y2}}{\sinh\prt{\frac{y+v}{2}}}.
    \end{align*}
\end{lemma}
\begin{proof}
    For brevity, we analyse $Q_{n,m}(\ell,y,v)^2$.
    Thus, by Definition \ref{def: w and W}, we have that
    \begin{align}\label{eq: Q squared nice}
        Q_{n,m}(\ell,y,v)^2 &= \frac{\cosh (y+n\ell)-\cosh n\ell}{\cosh (y+v+m\ell)-\cosh m\ell}\frac{\cosh m\ell -1}{\cosh n\ell -1}.
    \end{align}
    By the definition of hyperbolic cosine,
    \begin{align*}
        Q_{n,m}(\ell,y,v)^2 &=\frac{e^{y+n\ell}+e^{-y-n\ell}-e^{n\ell}-e^{-n\ell}}{e^{y+v+m\ell}+e^{-y-v-m\ell}-e^{m\ell}-e^{-m\ell}}\cdot \frac{e^{m\ell}+e^{-m\ell} - 2}{e^{n\ell}+e^{-n\ell} - 2}.
    \end{align*}
    Taking out common factors of $e^{n\ell}$ and $e^{m\ell}$, and letting $\ell\geq \eps >0$,
    \begin{align*}
        Q_{n,m}(\ell,y,v)^2 =\frac{e^y - 1 - e^{-2n\ell}(1-e^{-y})}{e^{y+v}-1 -e^{-2m\ell}(1-e^{-y-v})}\cdot\prt{1+O_\eps \prt{e^{-\ell\min \{ n,m\}}}}.
    \end{align*}
    Taking out a common factor of $(e^y - 1)/(e^{y+v}-1)$
    \begin{align*}
        Q_{n,m}(\ell,y,v)^2 &= \frac{e^y - 1}{e^{y+v}-1} \frac{1-e^{-y-2n\ell}}{1-e^{-y-v-2m\ell }}\prt{1+O_\eps \prt{e^{-\ell\min \{ n,m\}}}}\\
        &=\frac{e^y - 1}{e^{y+v}-1}\prt{1+O_\eps \prt{e^{-\ell\min \{ n,m\}}}}.
    \end{align*}
    Taylor expanding $\sqrt{1+x}$ around $x=0$ implies the first statement of the lemma.
    
    The $\ell \downarrow 0$ leading order behaviour can immediately be deduced from Taylor expansions of (\ref{eq: Q squared nice}), with
    \begin{align*}
        Q_{n,m}(\ell,y,v)^2\underset{\ell\downarrow 0}{\sim}\frac{\cosh y -1}{\cosh (y+v) - 1} \frac{m^2}{n^2}.
    \end{align*}
    Using the identity $\cosh x - 1 = 2\sinh^2(x/2)$ one arrives at the second statement of the lemma.
\end{proof}
With these results at hand, we finish this section with two lemmas aimed at understanding the behaviour of $A_{n,m}$ for large and small $\ell$. These are Lemmas \ref{lem: large l asymps of Anm} and \ref{lem: small l asymps of Anm}, respectively.
\begin{lemma}[Large $\ell$ asymptotics of $A_{n,m}$]\label{lem: large l asymps of Anm}
    Let $\eps>0$. Uniformly in $y,v>0$ and $n,m\geq 1$, for any $\ell\geq \eps$,
    \begin{equation*}
        A_{n,m}(\ell,y,v) = \frac{2K\prt{\sqrt{\frac{e^{y}-1}{e^{y+v}-1}}}}{\sqrt{e^{y+v}-1}}e^{-\ell \frac{n+m}{2}}\prt{1+ O_\eps(e^{-\ell \min \{ n,m\}})}.
    \end{equation*}
\end{lemma}
\begin{proof}
    Recall the definition of $A_{n,m}(\ell,y,v)$
    \begin{align*}
        A_{n,m}(\ell,y,v)&: = \1_{Q_{n,m}(\ell,y,v)<1}\cdot  \frac{K\prt{Q_{n,m}(\ell,y,v)}}{\sqrt{\cosh (y+v+m\ell)-\cosh m\ell}\sqrt{\cosh n\ell -1}} \\
        & + \1_{Q_{n,m}(\ell,y,v)>1}\cdot \frac{K\prt{\frac{1}{Q_{n,m}(\ell,y,v)}}}{\sqrt{\cosh (y+n\ell)-\cosh n\ell}\sqrt{\cosh m\ell -1}}.
    \end{align*}
    Fix $\eps>0$, and we assume $\ell\geq \eps$. By Lemma \ref{lem: asymp of Q}, it follows that for $\ell$ large enough, $Q_{n,m}(\ell,y,v)< 1$. Thus, to prove the $\ell\to \infty$ asymptotics, it suffices to analyse the first summand in the definition of $A_{n,m}$ above.

    We thus analyse the asymptotic behaviour of the denominator first. Let $x := y+v$. Then, 
    \begin{equation*}
        \frac{1}{\sqrt{\cosh(x+m\ell)-\cosh m\ell}\sqrt{\cosh n\ell - 1}} = \frac{1}{\sqrt{\frac{e^{x+m\ell}+e^{-x-m\ell}}{2}-\frac{e^{m\ell}+e^{-m\ell}}{2}}\sqrt{\frac{e^{n\ell }+e^{-n\ell}}{2}-1}}.
    \end{equation*}
    Removing common factors of $e^{m\ell}$ and $e^{n\ell}$ from each factor in the denominator yields the asymptotic
    \begin{equation*}
        \frac{2}{\sqrt{e^{y+v} - 1}}e^{-\frac{n+m}{2}}\prt{1+O_\eps(e^{-\ell \min \{ n,m\}})}.
    \end{equation*}
    
    Via Taylor expanding $K(Q_{n,m}(\ell,y,v))$ (c.f.\ Lemma \ref{lem: props of K}), using Lemma \ref{lem: asymp of Q} on the $\ell\to \infty$ asymptotics of $Q_{n,m}$, and the above asymptotic for the denominators, the lemma statement follows.
\end{proof}
\begin{lemma}[Small $\ell$ asymptotics of $A_{n,m}$]\label{lem: small l asymps of Anm}
    Let $n,m\geq 1$ and $y,v>0$. Then $A_{n,m}(\ell,y,v)$ has the following first-order asymptotics as $\ell\downarrow 0$
    \begin{equation*}
        A_{n,m}(\ell,y,v) \underset{\ell\downarrow0}{\sim} \begin{cases}
            \frac{\sqrt2}{n\sqrt { \ell}}\frac{K\prt{\frac{m}{n}\frac{\sinh \prt{\frac y2}}{\sinh \prt{\frac {y+v}2}}}}{\sqrt{\cosh(y+v)-1}} & \text{\emph{if} }\frac{m}{n}\frac{\sinh \prt{\frac y2}}{\sinh \prt{\frac {y+v}2}}< 1\\
            \frac{\sqrt 2}{m\sqrt{\ell}}\frac{K\prt{\frac{n}{m}\frac{\sinh \prt{\frac{y+v}{2}}}{\sinh \prt{\frac y2}}}}{\sqrt{\cosh(y)-1}}& \text{\emph{otherwise}}.
        \end{cases}
    \end{equation*}
\end{lemma}
\begin{proof}
    By Lemma \ref{lem: asymp of Q}, 
    \begin{equation*}
        \lim_{\ell\downarrow 0} Q_{n,m}(\ell,y,v) < 1 
    \end{equation*}
    if and only if
    \begin{equation*}
        \frac{m}{n}\frac{\sinh \prt{\frac y2}}{\sinh \prt{\frac {y+v}2}} < 1.
    \end{equation*}
    
    Using the simple asymptotics
    \begin{equation*}
        \frac{1}{\sqrt{\cosh n\ell - 1}}\underset{\ell \downarrow 0}\sim \frac{\sqrt 2}{n\sqrt{\ell}},\quad \frac{1}{\sqrt{\cosh (x+n\ell)-\cosh n\ell}}\underset{\ell\downarrow0}\sim \frac{1}{\sqrt{\cosh x - 1}},
    \end{equation*}
    together with Lemma \ref{lem: asymp of Q} and the definition of $A_{n,m}$, Definition \ref{def: A and Q}, the lemma follows.
\end{proof}
\begin{remark}
    The case when 
    \begin{equation*}
        \frac{m}{n}\frac{\sinh \prt{\frac y2}}{\sinh \prt{\frac {y+v}2}}= 1
    \end{equation*}
    would yield, by Lemma \ref{lem: props of K}, a logarithmic divergence as $\ell\downarrow 0$. As we shall only be integrating the function $A_{n,m}(\ell,y,v)$, we omit this case, since it would remain integrable around $\ell = 0$.
\end{remark}
\subsubsection{Development of $\hh'\cdot \hh'$}\label{subsubsec: development of hh prime squared}
As earlier remarked (c.f.\ Lemma \ref{lem: form of Hnm in terms of Anm}), we develop the form of $\hh'(\cdot) \cdot \hh'(\cdot)$ here.

Recall by (\ref{eq: hhat form}),
\begin{equation*}
    \hh(y) = \frac{2}{L}\cos(\tau y )\hf \prt{\frac{y}{L}}.
\end{equation*}

Differentiating $\hh$ one gets
\begin{equation*}
    \hh'(y) = \frac{2}{L}\prt{-\tau\sin(\tau y)\hf\prt{\frac{y}{L}}+\frac 1 L\cos(\tau y)\hf'\prt{\frac{y}{L}}} .
\end{equation*}

Hence, when we multiply the two $\hh'$ and group together the cross terms, we get
\begin{equation}\label{eq: h'*h' in terms of Ti}
    \hh'(y_1)\hh'(y_2) = \frac{4}{L^2}\sum_{i=0}^2 \frac{(-1)^{i}\tau^{2-i}}{L^{i}} T_{i}(y_1,y_2),
\end{equation}
where
\begin{align}\label{eq: T_0 T_1 T_2 explicitly}
    \begin{split}
        T_0(y_1,y_2)&:= \sin(\tau y_1)\sin(\tau y_2)\hf \prt{\frac{y_1}{L}}\hf \prt{\frac{y_2}{L}},\\
    T_1(y_1,y_2)&:=\cos(\tau y_1)\sin(\tau y_2) \hf' \prt{\frac{y_1}{L}}\hf \prt{\frac{y_2}{L}}+\sin(\tau y_1)\cos(\tau y_2)\hf \prt{\frac{y_1}{L}}\hf' \prt{\frac{y_2}{L}},\\
    T_2(y_1,y_2)&:= \cos(\tau y_1)\cos(\tau y_2)\hf' \prt{\frac{y_1}{L}}\hf '\prt{\frac{y_2}{L}}.
    \end{split}
\end{align}
In particular, for $i\in \{0,1,2\}$ 
\begin{equation}\label{eq: T_i are unfiformly upper bounded}
    |T_i| \leq \|\hf\|_{\infty}^2 \ll_{f}1.
\end{equation}

Using the product to sum formulae of trigonometric functions we have e.g.\ 
\begin{align}\label{eq: linear_Ti after product to sum}
    \begin{split}
        T_0(y_1,y_2) = &\frac12 \prt{\cos (\tau(y_1-y_2))- \cos(\tau(y_1+y_2))}\hf \prt{\frac{y_1}{L}}\hf \prt{\frac{y_2}{L}}.
    \end{split}
\end{align}
\subsection{Asymptotics of main integral}\label{subsec: asymptotics of main integral}
The goal of this section is to prove Proposition \ref{prop: asymp of relevant H integral}.

Recall by (\ref{eq: h'*h' in terms of Ti}) that we write
\begin{equation*}
    \hh'(y+n\ell)\hh'(y+v+m\ell) = \frac{4}{L^2}\sum_{i=0}^2 \frac{(-1)^i \tau^{2-i}}{L^i}T_i(y+n\ell,y+v+m\ell).
\end{equation*}
Further recall by Lemma \ref{lem: form of Hnm in terms of Anm} that
\begin{equation*}
    H(\ell) = \frac{4\ell}{\pi^2}\sum_{n,m\geq 1} \izi \izi \hh' (y+n\ell) \hh' (y+v+m\ell) A_{n,m}(\ell,y,v)\d y \td v.
\end{equation*}
We thus define, for integers $n,m\geq 1$, $\ell>0$, and $i\in \{0,1,2\}$, the function 
\begin{equation}\label{eq: def of F nm i}
    F_{n,m}^{(i)}(\ell) := (-1)^i\frac{16 \tau^{2-i}\ell}{\pi^2 L^{2+i}} \izi \izi T_i(y+n\ell,y+v+m\ell)A_{n,m}(\ell,y,v)\d y\td v.
\end{equation}
Then
\begin{equation}\label{eq: H in terms of F nm i}
    H(\ell) = \sum_{i=0}^2\sum_{n,m\geq 1} F_{n,m}^{(i)}(\ell). 
\end{equation}

Proposition \ref{prop: asymp of relevant H integral} follows from combining the following three lemmas.

\begin{lemma}\label{lem: F110}
For fixed $\tau>0, f$ and $L\geq 1$,
    \begin{equation*}
        \izi F_{1,1}^{(0)}(\ell) \frac{\sinh ^2 \prt{\frac\ell2}}{\ell}\d \ell = \frac{\tau \tanh\pi \tau}{4L}\|f\|^2_{L^2(\R)} + O_{f,\tau}\prt{\frac{1}{L^2}}.
    \end{equation*}
\end{lemma}
\begin{lemma}\label{lem: F111 and F112}
    For fixed $\tau>0, f$ and $L\geq 1$
    \begin{equation*}
        \izi \prt{\left|F_{1,1}^{(1)}(\ell)\right| +\left|F_{1,1}^{(2)}(\ell)\right|  }\frac{\sinh ^2 \prt{\frac\ell2}}{\ell}\d \ell = O_{f,\tau}\prt{\frac 1{L^2}}.
    \end{equation*}
\end{lemma}
\begin{lemma}\label{lem: sinh squared Anm integral n+m geq 3}
    The following expression
    \begin{equation*}
        \sum_{n+m\geq 3} \izi \izi\izi \sinh^2 \prt{\frac\ell 2} A_{n,m}(\ell,y,v)\d y\td v\td\ell,
    \end{equation*}
    converges.
\end{lemma}
We now proceed to prove Proposition \ref{prop: asymp of relevant H integral} using these three lemmas. Afterwards, we end the section by proving them.
\begin{proof}[Proof of Proposition \ref{prop: asymp of relevant H integral}]
    By the absolute convergence in Corollary \ref{cor: H sinh square integral converges}, we write
    \begin{equation*}
        I:=\izi H(\ell) \frac{\sinh^2 \prt{\frac \ell2}}{\ell}\d\ell = \sum_{i=0}^2 \sum_{n,m\geq 1} \izi F_{n,m}^{(i)}(\ell) \frac{\sinh^2 \prt{\frac\ell2}}{\ell}\d\ell.
    \end{equation*}
    By Lemmas \ref{lem: F110} and \ref{lem: F111 and F112} we have
    \begin{equation*}
        I = \frac{\tau \tanh\pi \tau }{4L }\|f\|_{L^2(\R)}^2 + \text{DoubSum} + O_{f,\tau}(L^{-2}),
    \end{equation*}
    where
    \begin{equation*}
        \text{DoubSum} := \sum_{i=0}^2 \sum_{n+m\geq 3} \izi F_{n,m}^{(i)}(\ell) \frac{\sinh^2\prt{\frac\ell2}}{\ell}\d\ell.
    \end{equation*}
    We only require bounding the double sum term to prove the proposition.
    By the definition of $F_{n,m}^{(i)}$, as each $T_i$ is uniformly upper bounded by $\|\hf\|_{\infty}^2$ (\ref{eq: T_i are unfiformly upper bounded}), we have that
    \begin{equation*}
        \text{DoubSum} \ll_{f,\tau} \frac{1}{L^2} \sum_{n+m\geq 3} \izi \izi \izi A_{n,m}(\ell,y,v)\sinh^2\prt{\frac \ell 2}\d y\td v\td \ell.
    \end{equation*}
    By Lemma \ref{lem: sinh squared Anm integral n+m geq 3}
    \begin{equation}
        \text{DoubSum} \ll_{f,\tau}\frac 1{L^2},
    \end{equation}
    which proves the proposition.
\end{proof}
We now proceed to prove the three lemmas above.
\begin{proof}[Proof of Lemma \ref{lem: F110}]
    By definition of $F_{1,1}^{(0)}$,
    \begin{align*}
        I&:= \izi F_{1,1}^{(0)}(\ell) \frac{\sinh^2 \prt{\frac \ell 2}}{\ell}\d\ell\\
        &= \frac{16\tau^2}{\pi^2L^2}  \izi\izi\izi \sinh^2\prt{\frac \ell2} T_0(y+\ell,y+v+\ell) A_{1,1}(\ell,y,v)\d y \td v \td \ell.
    \end{align*}

    Anticipating the use of (\ref{eq: linear_Ti after product to sum}), we define for $\ell,y,v>0$,
    \begin{equation*}
        U(\ell,y,v) := \frac{8\tau^2}{\pi^2 L^2}\sinh^2\prt{\frac \ell 2}A_{1,1}(\ell,y,v)\hf\prt{\frac{y+\ell}{L}}\hf\prt{\frac{y+v+\ell}{L}}.
    \end{equation*}
    By the absolute convergence of the integrals, and (\ref{eq: linear_Ti after product to sum}), we let
    \begin{align*}
        I = I_1-I_2
    \end{align*}
    where
    \begin{align*}
        I_1&:= \izi\izi\izi \cos(\tau v)U(\ell,y,v)\d\ell \td  y\td v,\\
        I_2&:= \izi\izi\izi \cos(\tau (v+2y+2\ell)) U(\ell,y,v) \d\ell \td y\td v.
    \end{align*}
    We defer proving that $I_2 = O_{f,\tau}(L^{-2})$ to the end.

    We analyse the large $L$ behaviour of $I_1$. First, we change variables $\ell\mapsto L\ell$, such that
    \begin{align*}
        I_1 &= L \izi\cos(\tau v) \izi\izi U(L\ell,y,v)\d\ell \td y \td v\\
        &=\frac{8\tau^2}{\pi^2 L}\izi \cos(\tau v) \izi \izi \sinh^2\prt{\frac{L\ell}{2}}A_{1,1}(L\ell,y,v)\hf \prt{\ell+\frac yL}\hf\prt{\ell + \frac{y+v}L}\td\ell\td y\td v.
    \end{align*}
    
    We recall by Lemma \ref{lem: large l asymps of Anm} that, for any $\eps>0$
    \begin{equation*}
        A_{1,1}(\ell,y,v) = \frac{2K\prt{\sqrt{\frac{e^y-1}{e^{y+v}-1}}}}{\sqrt{e^{y+v}-1}}e^{-\ell}\prt{1+O_\eps(e^{-\ell})}
    \end{equation*}
    for $\ell\geq \eps$. 
    We may remove the non-uniformity of the asymptotics of $A_{1,1}$ at $\ell = 0$, i.e.\ drop the $\eps$. Indeed, the $\ell$-integral above converges, by the square-root divergence at $\ell = 0$ of $A_{1,1}$ (c.f.\ Lemma \ref{lem: small l asymps of Anm}) and the compact support of $\hf$. Thus,
    \begin{align*}
        I_1 = \frac{16\tau^2}{\pi^2 L}&\izi \cos(\tau v)\izi\frac{K\prt{\sqrt{\frac{e^y-1}{e^{y+v}-1}}}}{\sqrt{e^{y+v}-1}}\\&\izi \sinh^2\prt{\frac{L\ell}2} e^{-L\ell}\hf \prt{\ell+\frac yL}\hf\prt{\ell + \frac{y+v}L}\prt{1+O(e^{-L\ell})} \d\ell \td y \td v.
    \end{align*}
    The $O(e^{-L\ell})$ error is with respect to $L$, and shall remain so throughout this proof.
    
    As $\sinh^2 (L\ell/2) = \frac 14 \prt{e^{L\ell}+e^{-L\ell}-2}$,
    \begin{align}\label{eq: I_1 for I_2}
        \begin{split}
            I_1=\frac{4\tau^2}{\pi^2 L}&\izi \cos(\tau v)\izi\frac{K\prt{\sqrt{\frac{e^y-1}{e^{y+v}-1}}}}{\sqrt{e^{y+v}-1}}\\&\izi\hf \prt{\ell+\frac yL}\hf\prt{\ell + \frac{y+v}L}\prt{1+O(e^{-L\ell})} \d\ell \td y \td v.
        \end{split}
    \end{align}
    Since $\hf \in C_c^\infty(\R)$ with $\supp \hf \subseteq [-1,1]$, the third integral goes from $\ell = 0$ to $\ell =1$. Furthermore, uniformly in $\ell\in [0,1]$, we have for $0\leq x\leq L$,
    \begin{equation*}
        \hf\prt{\ell+\frac{x}{L}} = \hf(\ell) +O_f\prt{\frac{x}{L}}.
    \end{equation*}
    Thus, 
    \begin{align*}
        I_1 &= \frac{4\tau^2}{\pi^2L}\izi \cos (\tau v) \izi \frac{K\prt{\sqrt{\frac{e^y-1}{e^{y+v}-1}}}}{\sqrt{e^{y+v}-1}}\int_0^1 f(\ell)^2+O_f((y+v)e^{-L\ell})\d\ell \td y \td v.
    \end{align*}

    The following ingredients combine to yield that
    \begin{equation*}
        I_1 = \frac{\tau \tanh\pi\tau}{4\pi L}\|f\|_{L^2(\R)}^2 +O_{f,\tau}\prt{\frac 1{L^2}}.
    \end{equation*}
    \begin{itemize}
        \item The function $\hf$ has $\supp \hf \subseteq [-1,1]$, thus $ \int_0^1 f(\ell)^2\d\ell = \|\hf\|_{L^2(\R_{\geq 0})}^2$. Furthermore, as $\hf$ is even, $\|\hf\|_{L^2(\R_{\geq 0})}^2 = \frac 12 \|\hf \|_{L^2(\R)}^2$. Lastly, the Fourier transform is not normalised (c.f.\ Definition \ref{def: fourier transform}). Thus Plancherel's identity yields
        \begin{equation*}
            \int_0^1 \hf(\ell)^2 = \frac{1}{4\pi}\|f\|_{L^2(\R)}^2.
        \end{equation*}
        \item One has that
        \begin{equation*}
            \int_0^1 O_f(e^{-L\ell})\d\ell = O_f\prt{\frac{1}{L}}.
        \end{equation*}
        Furthermore, by the exponential decay in both the $y,v$ variables
        \begin{equation*}
            \izi (y+v)\cos(\tau v)\izi \frac{K\prt{\sqrt{\frac{e^y-1}{e^{y+v}-1}}}}{\sqrt{e^{y+v}-1}}\d y\td v \ll_{\tau} 1.
        \end{equation*}
        \item Lemma \ref{lem: linear_cosine K integral} evaluates the double integral explicitly as
        \begin{equation*}
            \izi \cos (\tau v) \izi \frac{K\prt{\sqrt{\frac{e^y-1}{e^{y+v}-1}}}}{\sqrt{e^{y+v}-1}}\d y\td v = \frac{\pi^2}{4\tau}\tanh\pi\tau.
        \end{equation*}
    \end{itemize}

    What now remains is to prove that $I_2 = O_{f,\tau}(L^{-2})$. By the exact same arguments as for $I_1$ (c.f.\ (\ref{eq: I_1 for I_2})),
    \begin{align*}
        I_2 =\frac{4\tau^2}{\pi^2 L}&\izi\izi\frac{K\prt{\sqrt{\frac{e^y-1}{e^{y+v}-1}}}}{\sqrt{e^{y+v}-1}}\izi \cos(\tau (v+2y+2L\ell))\\&\hf \prt{\ell+\frac yL}\hf\prt{\ell + \frac{y+v}L}\prt{1+O(e^{-L\ell})} \d\ell \td y \td v.
    \end{align*}
    Integration by parts in the $\ell$-variable yields
    \begin{align*}
            &\frac 1L\izi \cos(\tau(2y+v+2L\ell))\hf\prt{\ell +\frac{y}{L}}\hf\prt{\ell +\frac{y+v}{L}}\d\ell\\
            &=  -\frac{1}{2\tau L^2} \sin(\tau(2y+v)) \, \hf\prt{\frac{y}{L}} \hf\prt{\frac{y+v}{L}}\\& - \frac{1}{2\tau L^2} \int_{0}^{\infty} \sin(\tau(2y+v+2L\ell)) \left[ \hf'\prt{\ell +\frac{y}{L}} \hf\prt{\ell +\frac{y+v}{L}} + \hf\prt{\ell +\frac{y}{L}} \hf'\prt{\ell +\frac{y+v}{L}} \right] \d\ell.
        \end{align*}
        Together with Lemma \ref{lem: linear_cosine K integral},
        \begin{equation*}
            I_2 = O_{f,\tau}\prt{\frac 1{L^2}}.
        \end{equation*}
\end{proof} 

\begin{proof}[Proof of Lemma \ref{lem: F111 and F112}]
    Recall by definition (\ref{eq: def of F nm i}) that
    \begin{align*}
        &I:=\izi \prt{\left|F_{1,1}^{(1)}(\ell)\right| +\left|F_{1,1}^{(2)}(\ell)\right|  }\frac{\sinh ^2 \prt{\frac\ell2}}{\ell}\d \ell \\ &= \frac 1{L^2}\sum_{i=1}^2 \frac{16\tau^{2-i}}{L^{i}}\izi\izi\izi \sinh^2\prt{\frac \ell2}|T_i(y+\ell,y+v+\ell)|A_{1,1}(\ell,y,v)\d y\td v\td\ell.
    \end{align*}
    The $i=1$ summand has a $1/L^3$ prefactor, and the $i=2$ summand has a $1/L^4$ prefactor. Thus, if one follows the same steps as in the above Lemma \ref{lem: F110}, the integrals converge, and each summand has a prefactor of at least $1/L^2$. This yields that
    \begin{equation*}
        I = O_{f,\tau}\prt{\frac 1{L^2}}.
    \end{equation*}
\end{proof}
\begin{proof}[Proof of Lemma \ref{lem: sinh squared Anm integral n+m geq 3}]
    Define
    \begin{equation*}
        I_{n,m}:=\izi\izi\izi\sinh^2\prt{\frac \ell2}A_{n,m}(\ell,y,v)\d y\td v\td \ell.
    \end{equation*}
    The goal of the lemma is to prove that $\sum_{n,m\geq 3}I_{n,m} $ converges absolutely.

    Firstly, $I_{n,m}$ converges, by the small $\ell$ and large $\ell$ behaviour of $\sinh^2(\cdot/2)A_{n,m}$ (c.f.\ Lemmas \ref{lem: large l asymps of Anm} and \ref{lem: small l asymps of Anm}). Furthermore, as $\sinh^2(\ell/2)$ has a double zero at $\ell =0$, the non-uniformity in $n,m$ of both of the aforementioned lemmas may be dropped. Thus, by plugging in Lemma \ref{lem: large l asymps of Anm},
    \begin{equation*}
        I_{n,m}= \izi\izi\izi \sinh^2 \prt{\frac\ell 2} \frac{2K\prt{\sqrt{\frac{e^y - 1}{e^{y+v}-1}}}}{\sqrt{e^{y+v}-1}}e^{-\ell\frac{n+m}{2}}\prt{1+O\prt{e^{-\min\{ n,m\} \ell}}}\d\ell \td y \td v.
    \end{equation*}

    This integral is actually decoupled, and thus, by Lemma \ref{lem: linear_cosine K integral}, it suffices to prove that
    \begin{equation*}
        \sum_{n+m\geq 3} \izi \sinh^2 \prt{\frac\ell2} e^{-\ell\frac{n+m}{2}} \d\ell <\infty.
    \end{equation*}
    This follows by direct integration.
\end{proof}
\section{Length-minimising loops and their properties}\label{sec: len minimising loops}
It is well-known that a shortest noncontractible geodesic loop through a point (a.k.a.\ a \emph{systole}) on a closed hyperbolic surface \emph{is simple} (c.f.\ \cite[Lemma 4.1.5]{buser1992geometry}). This is an example of what we define as a \emph{length-minimising geodesic loop}, c.f.\ Definition \ref{def: length min loop}.

The main aim of this section is to prove general results that characterise the shape and topology of length-minimising loops and sequences (c.f.\ \ref{def: length min seq}), that will be useful in the sequel. In particular, \emph{all} length-minimising loops are simple.

For useful definitions, c.f.\ the prerequisites in Section \ref{subsec: topological definitions}.

This section is split as follows:
\begin{itemize}
    \item[(i)] Section \ref{subsec: ALA decomp}: defines a useful decomposition of non-simple arcs.
    \item[(ii)] Section \ref{subsec: len min loops}: defines length-minimising geodesic loops and proves that they are simple.
    \item[(iii)] Section \ref{subsec: len min seq}: defines length-minimising sequences and characterises the topology of the associated length-minimising loops.
    \item[(iv)] Section \ref{subsec: n exploring loops}: provides an example of a length-minimising sequence of loops that iteratively ``explores" the topology of the surface.
    \item[(v)] Section \ref{subsec: geom of two shortest loops}: collects standard results on the geometry around short geodesic loops, which will be used in the sequel.
\end{itemize}
\subsection{Arc-Loop-Arc decomposition}\label{subsec: ALA decomp}
    We begin with the following simple decomposition of a loop into arcs.
    \begin{numnotation}[Arc-Loop-Arc decomposition]\label{not: ALA decomp}
        Let $c:[0,1]\ra X$ be a non-simple arc with a finite number of self-intersections. Define the time of first self-intersection
        \begin{align*}
            t&:=\sup\{ t'>0:\forall s\in[0,t'), c(s)\neq c(t')\}
        \end{align*}
        and $s\in [0,t)$ be the unique time for which $c(s) = c(t)$.
        Then the \emph{arc-loop-arc decomposition} of $c$ (or the \emph{ALA-decomposition} of $c$ for short) is
        \begin{equation*}
            c = \eta_1\dot b\dot\eta_2,
        \end{equation*}
        where
        \begin{align*}
            \eta_1:= c_{|[0,s]}, \quad b := c_{|[s,t]}, \quad\eta_2:= c_{|[t,1]}.
        \end{align*}
    \end{numnotation}
    We note that the arc $\eta_1$ is simple, the loop $b$ is simple, and they trivially intersect. For an example, see Figure \ref{fig: ALAdecomp 2 selfintersections}.

    \begin{numnotation}[ALA-point]\label{not: ALA point}
        The ALA-decomposition of a non-simple arc also yields a common intersection point, which we call the \emph{ALA-point} 
        $$p:= e_1(\eta_1)=e_0(\eta_2)=e_0(b)=e_1(b).$$
    \end{numnotation}
    \begin{figure}
        \centering
        \includegraphics[width=0.45\linewidth]{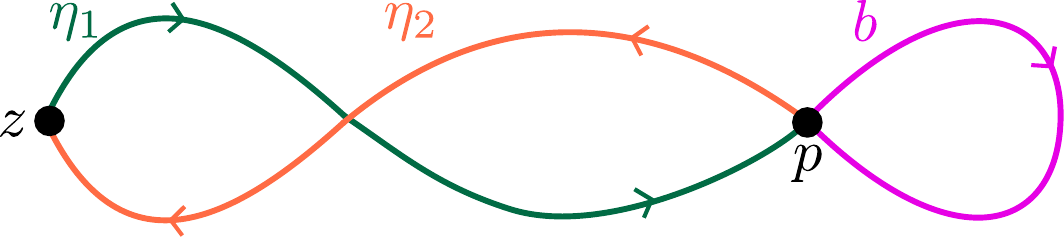}
        \caption{The ALA-decomposition of a loop with two self-intersections}
        \label{fig: ALAdecomp 2 selfintersections}
    \end{figure}
    \subsection{Length-minimising loops}\label{subsec: len min loops}
    From here on out, we fix some closed hyperbolic surface $X$ and a point $z\in X$. 
    \begin{numnotation}
        If we have a subset $G\subseteq \pi_1\xz$, and some loop $\delta$ based at $z$, then we say $\delta\in G$ if its homotopy class is in $G$. In other words, $\delta$ is homotopic to some element in $G$.
    \end{numnotation}
    \begin{definition}[Length-minimising loop]\label{def: length min loop}
        Let $G\lneq \pi_1\xz$ be a subgroup. We say a geodesic loop $\delta$ is \emph{length-minimising with respect to} $G$ if $\delta\not\in G$, and for all geodesic loops $\eta\in \pi_1\xz, \eta \not\in  G$
        \begin{equation*}
            \ell(\delta) \leq \ell(\eta).
        \end{equation*}
    \end{definition}
    \begin{remark}
        Any length-minimising loop is noncontractible as any subgroup must contain $\id$.
    \end{remark}
    A simple (and natural) example of such a length-minimising loop is a \emph{systole based at $z$}, which is a length-minimising loop with respect to the subgroup $G = \{\id\} \lneq \pi_1\xz$. By definition, a systole is a shortest noncontractible loop based at a point, and it is always simple (c.f.\ \cite[Lemma 4.1.5]{buser1992geometry}).
    
    Another example, which forms a central ingredient in our approach to proving Proposition \ref{prop: off diagonal term claim} in Section \ref{sec: offdiagterms}, is the notion of a \emph{second shortest primitive loop}.

    \begin{restatable}[Second shortest primitive loop]{definition}{defsecondshortestgeodesicloop}\label{def: second shortest primitive loop}
        Let $\g$ be a systole based at $z$, which we identify with its homotopy class. A geodesic loop $\delta$ based at $z$ is \emph{a second shortest primitive loop} with respect to $\g$ if it is length-minimising with respect to $G:=\{\g^m \}_{m\in \Z}$.
    \end{restatable}
    
    \begin{remark}
        The notion of second shortest primitive loop depends on the choice of systole, if the systole is not uniquely defined.
    \end{remark}
    We conclude this section by proving that any length-minimising loop has to be \emph{simple}.
    \begin{theorem}[Topology of length-minimising loops] \label{thm: top of len min}
        Let $\delta$ be a length-minimising loop w.r.t.\ a subgroup $G \lneq \pi_1\xz$. Then $\delta$ is simple.
    \end{theorem}
\begin{proof}

    Assume by way of contradiction that $\delta$ is not simple. We write $\delta$ in its ALA-decomposition as 
    \begin{equation*}
        \delta = \eta_1\dot b\dot \eta_2
    \end{equation*}
    with ALA-point $p$.

    Define $\eta:=\eta_1\dot\eta_2$. By definition $\ell(\eta)<\ell(\delta)$. Thus in particular $\ell_\xz(\eta) \leq \ell(\eta) < \ell(\delta)$, which implies by the definition of $\delta$ that $\eta\in G$.

    We first prove that $p\neq z$ by contradiction. Assume $p=z$. Then $b$ is a loop based at $z$, and thus by the length-minimality of $\delta$, $b\in G$. By the above, $\eta_2$, which here equals $\eta$, is also in $G$. Thus, as $G$ is closed under concatenation, $\delta = b\dot\eta_2\in G$, contradicting the definition of $\delta$.

    Define the loop 
    \begin{align*}
        \delta' := \begin{cases}
            \eta_1 \dot b \dot \bar\eta_1 & \ell(\eta_1 )\leq \ell(\eta_2) \\
            \bar\eta_2 \dot b\dot \eta_2 & \text{otherwise},
        \end{cases}
    \end{align*}
    c.f.\ Figure \ref{fig: len min figure eight} for an example when $\ell(\eta_1)\leq \ell(\eta_2)$.

    The loop $\delta'$ is not a geodesic loop. Indeed, $\delta'$ has an infinite amount of self-intersections along $\eta_1$, respectively $\eta_2$. Hence by Lemma \ref{lem: on geodesic loops}, $\ell_\xz(\delta') < \ell(\delta')$. Furthermore, $\ell(\delta')\leq \ell(\delta)$ by definition. Thus, by definition of $\delta$, $\delta' \in G$.
    
    We have
    \begin{equation*}
        \delta  = \eta_1 \dot b\dot \eta_2 \sim \begin{cases}
            \eta_1 \dot b \dot \bar\eta_1 \dot \eta_1 \dot \eta_2 =\delta' \dot \eta & \text{if } \ell(\eta_1)\leq \ell(\eta_2)\\
            \eta_1 \dot \eta_2 \dot \bar\eta_2 \dot b \dot \eta_2 = \eta \dot \delta' & \text{otherwise.}
        \end{cases}
    \end{equation*}
    In both cases, as $\delta'$ and $\eta $ are in $G$, $\delta\in G$. This contradicts the definition of $\delta$.

    \begin{figure}[h!]
    \centering
    \begin{subfigure}[t]{0.4\textwidth}
        \centering
        \includegraphics[width=0.7\linewidth]{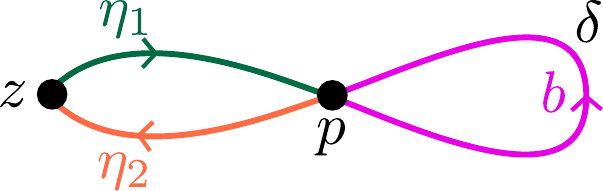}
        \caption{ALA-decomposition of $\delta$.}
        \label{fig: aladecompfigeight}
    \end{subfigure}
    ~
    \begin{subfigure}[t]{0.4\textwidth}
        \centering
        \includegraphics[width=0.7\linewidth]{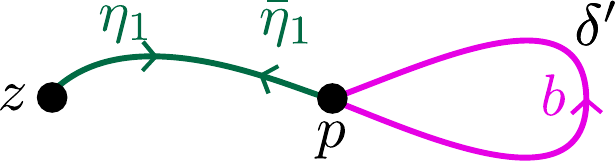}
        \caption{The loop $\delta'$.}
        \label{fig: lenmindeltaprime}
    \end{subfigure}
    \caption{Example of the loops appearing in the proof, when $\delta$ has one self-intersection.}
    \label{fig: len min figure eight}
\end{figure}

\end{proof}
\subsection{Length-minimising sequence}\label{subsec: len min seq}
We continue by defining appropriate \emph{sequences} of length-minimising loops based at a point. This generalises the notion of a second shortest primitive loop.
\begin{definition}[Length-minimising sequence]\label{def: length min seq}
    A sequence of (subgroup, geodesic-loop) pairs $(\{\id\}=G_0,\id ) , (G_1,\g_1),\ldots, (G_n,\g_n)$ is called a \emph{length-minimising sequence of length $n\geq 1$} if for all $i\in \{1,\ldots, n\}$, 
    \begin{itemize}
        \item[] (i) $G_{i-1}\lneq G_i \leq \pi_1\xz$;
        \item[] (ii) $\g_i\in G_i$; 
        \item[] (iii) the geodesic loop $\g_i$ is length-minimising w.r.t.\ $G_{i-1}$.
    \end{itemize}
\end{definition}
We prove that the associated length-minimising geodesic loops \emph{intersect trivially} (i.e.\ only at times $0$ and $1$), which is the subject of the following theorem.
\begin{theorem}\label{thm: len min seq simplicity}
    Let $(G_1,\g_1),\ldots, (G_n,\g_n)$ be a length-minimising sequence of length $n$. Then for any $1 \leq i <  j \leq n$, the geodesic loops $\g_i$ and $\g_j$ intersect trivially.
\end{theorem}
\begin{proof}
    Assume by way of contradiction that there exist $1\leq i < j \leq n$ for which $\g_i$ and $\g_j$ intersect nontrivially. Let $\g := \g_i$ and $\delta := \g_j$.

    First, $\g$ and $\delta$ do not intersect in an interval. If they did, one has by extending the interval of agreement between the two  that $\g= \delta$, as they are simple (by Theorem \ref{thm: top of len min}).
    
    Thus, if they intersect, they do so transversely and at finitely many points. Let $0< t_1 < \cdots < t_I < 1$ be the nontrivial times of intersection of $\delta$ with $\g$.

    Split $\delta = \beta_1\dot\beta_2$ where $\beta_1 = \delta_{|[0,t_1]}, \beta_2 = \delta_{|[t_1,1]}$, and denote $p:=e_1(\beta_1) = e_0(\beta_2)$. Up to orientation we may assume $\ell(\beta_1) \leq \frac{\ell(\delta)}{2}$.

    Split $\g = \alpha_1 \dot\alpha_2$ with $\alpha_1$ being the first time $\g$ intersects $\delta$ at $p$. Again, up to orientation, we may assume $\ell(\alpha_1) \leq \frac{\ell(\g)}{2}$.

    By definition, we have that the loop $\delta_1 := \beta_1\dot\bar \alpha_1$ has length
    $$\ell(\delta_1) = \ell(\beta_1) + \ell(\alpha_1) \leq \frac{\ell(\delta)}{2}+\frac{\ell(\g)}{2}\leq \ell(\delta).$$
    
    One has that $\delta_1$ is homotopic to an element in $G_{j-1}$. Indeed, $\delta_1$ has a non-smooth point (as $p$ is a transverse intersection point) outside of its endpoint, which, by Lemma \ref{lem: on geodesic loops}, implies $\ell_\xz(\delta_1)<\ell(\delta)$. Thus by assumption on $\delta$, we have that $\delta_1 \in G_{j-1}$.

    Now define the loop $\delta_2 := \alpha_1 \dot\beta_2$. We have two options - either $\ell(\delta_2) \leq \ell(\delta)$ or $\ell(\delta_2)   > \ell(\delta)$. See Figure \ref{fig: combined} for an example of how we build these loops.
    \begin{figure}[h] 
    \centering 
        
        \begin{subfigure}[c]{0.49\textwidth}
            \centering
                \includegraphics[width=0.575\linewidth]{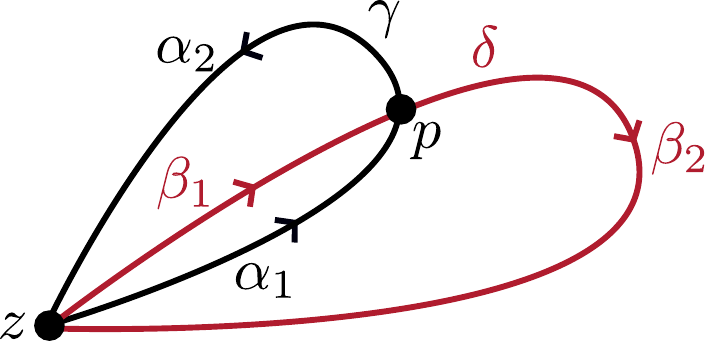}
                \caption{Here $\ell(\beta_1) < \ell(\alpha_1)$.}
                \label{fig: len min gamma delta}
        \end{subfigure}
        ~
        \begin{subfigure}[c]{0.25\textwidth}
            \centering
                \includegraphics[width=0.7\linewidth]{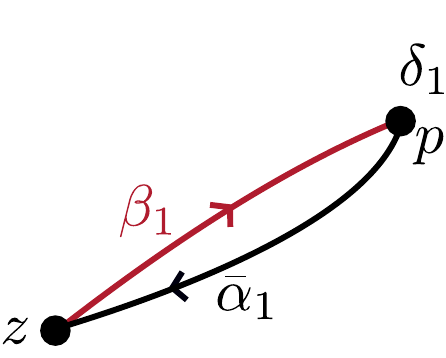}
                \caption{The loop $\delta_1$.}
                \label{fig:delta_1}
        \end{subfigure}
        ~
        \begin{subfigure}[c]{0.25\textwidth}
            \centering
                \includegraphics[width=0.7\linewidth]{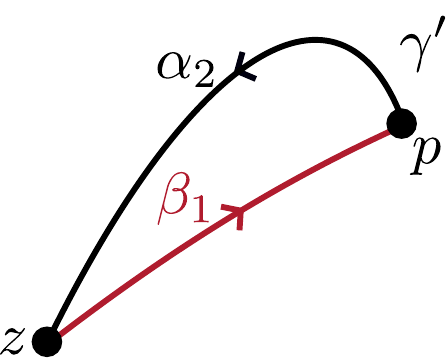}
                \caption{The loop $\g'$.}
                \label{fig: delta_2}
        \end{subfigure}
        \caption{Length-minimising loops $\g$ and $\delta$ intersecting at $p$. The contradiction follows as both $\delta_1$ and $\g'$ (Figures \ref{fig:delta_1} and \ref{fig: delta_2} respectively) have lengths $<\ell(\g)$, and $\g \sim \bar \delta_1 \dot \g'$.}
        \label{fig: combined}
        \end{figure}
    
    In the first case, if $\ell(\delta_2)\leq \ell(\delta)$ then, for the same reasons as for $\delta_1$, we have that $\delta_2\in G_{j-1}$. Thus
    \begin{equation*}
        \delta = \beta_1\dot\beta_2 \sim \beta_1 \dot\bar\alpha_1 \dot\alpha_1\dot \beta_2 = \delta_1\dot\delta_2\in G_{j-1},
    \end{equation*}
    contradicting its definition.

    In the second case, assume $\ell(\delta_2)=\ell(\alpha_1) +\ell(\beta_2) > \ell(\delta)$. In particular, this implies 
    \begin{equation*}
        \ell(\alpha_1) > \ell(\delta) -\ell(\beta_2) = \ell(\beta_1).
    \end{equation*}
    Hence the loops $\delta_1= \beta_1\dot\bar\alpha_1$ and $\g':= \beta_1\dot \alpha_2$ both have lengths $< \ell(\g)$. As such, they are both homotopic to elements in $G_{i-1}$. Thus
    \begin{equation*}
        \g = \alpha_1\dot \alpha_2 \sim \alpha_1\dot\bar\beta_1\dot \beta_1 \dot\alpha_2 = \bar \delta_1\dot\g'.
    \end{equation*}
    Therefore $\g\in G_{i-1}$, contradicting its definition.

    We have reduced all possibilities to contradictions, and thus we are done. 
\end{proof}
The above result, together with Theorem \ref{thm: top of len min}, immediately imply the following corollary. Furthermore, Theorem \ref{thm: second shortest simp and triv} follows as well.
\begin{corollary}\label{cor: length-minimising topology}
    For any length-minimising sequence $(G_1,\g_1),\ldots, (G_n,\g_n)$ with $n\geq 2$ one has:
    \begin{itemize}
        \item the geodesic loop $\g_1$ is a systole;
        \item each geodesic loop $\g_i$ is simple, and intersects every $\g_j, j\neq i$, trivially.
    \end{itemize}
\end{corollary}

We end by noting that the notion of length-minimising sequences allows for context-dependent definitions of what being ``the $n$-th shortest loop" means. Then Corollary \ref{cor: length-minimising topology} immediately applies, and tells us that these loops are simple and trivially intersect.
\subsection{$n$-exploring loops}\label{subsec: n exploring loops}
We introduce a certain length-minimising sequence that we believe might be of independent interest, which we call \emph{$n$-exploring loops}. This sequence ``explores" the topology of the surface, in the sense of Theorem \ref{thm: euler char of exploring}. Thanks to Theorems \ref{thm: top of len min} and \ref{thm: len min seq simplicity}, we know that each such loop is simple.

Before proceeding with the definition of exploring loops, we need the following notion of a \emph{weakly filled surface}.

\begin{definition}[Weakly filled surface]\label{def: weakly filled surface}
    Let $X$ be a compact hyperbolic surface and $\ug  := (\g_1,\ldots,\g_k)$ be a $k$-tuple of pairwise distinct geodesic loops based at a point $z$.
 
    Then the topological surface \emph{$\Sigma(\ug)$ weakly filled by $\ug$} is built by the following process.
    \begin{itemize}
        \item[1.] Let $\cN_\eps(\ug):= \{ w\in X: d(w,\ug)<\eps\}$ be a regular neighbourhood of $\ug$, where $\eps>0$ is small enough such that $\cN_\eps(\ug)$ retracts to $\ug$. 
        \item[2.] The topological surface $X - \cN_\eps(\ug)$ has $\bigsqcup_{i\in I}C_i$ connected components. We set
        \begin{equation*}
            \Sigma(\ug):=\cN_\eps(\ug) \cup \bigcup_{i\in I_0}C_i,
        \end{equation*}
        where $I_0 = \{ i\in I: C_i \text{ is a topological disk or cylinder}\}$.
    \end{itemize}
\end{definition}

     \emph{Example.} If $\g$ is a simple noncontractible loop, then $\Sigma(\g)$ is a cylinder. If $\g$ has a single self-intersection, then $\Sigma(\g)$ is a topological pair of pants.
     
The weakly filled surface $\Sigma(\ug)$ is connected, since by definition each disk and cylinder is bounded by $\partial \cN_\eps(\ug)$. Furthermore, the specific choice of a small enough $\eps>0$ is topologically inconsequential, as there is an isotopy fixing $T$ from one weakly filled surface to another.

Let $\iota:\Sigma(\ug)\hookrightarrow X$ be the inclusion map. The associated group homomorphism
\begin{equation*}
    \iota_*:\pi_1(\Sigma(\ug),z)\ra \pi_1\xz
\end{equation*}
is injective, as $\Sigma(\ug)$ contains every disk in $X - \cN_\eps(\ug)$, and hence every noncontractible loop in $\Sigma(\ug)$ is noncontractible in $X$.

We say $G(\ug):=\iota_*(\pi_1(\Sigma(\ug),z))\leq \pi_1\xz$ is the subgroup of \emph{$\ug$-discovered loops}. 
\begin{definition}[Exploring sequence]\label{def: exploring loops}
    We define an \emph{exploring sequence} by induction. Let $T_0 = \id$. For $n\geq 1$, if $G(T_{n-1}) \neq \pi_1\xz$, we let $\g_n$ be a length-minimising geodesic loop w.r.t.\ $G(T_{n-1})$ and define $T_n = T_{n-1} \cup \g_n$. Then the length $n$ \emph{exploring sequence} is 
    \begin{equation*}
        ((G(T_1),\g_1), \ldots, (G(T_n),\g_n)    ).
    \end{equation*}
    We call $\g_n$ an \emph{$n$-exploring loop}.
\end{definition}
\emph{Example.} A 2-exploring loop is a second shortest primitive loop, since any loop not homotopic in $X$ to a power of the systole $\g_1$ must leave the cylinder $\Sigma(\g_1)$ of the systole, which has fundamental group $\pi_1(\Sigma(\g_1),z)\cong \{\g_1^m\}_{m\in \Z}$.
    
We end by proving Theorem \ref{thm: euler char of exploring}, which motivates our particular definition of exploring loops.
\begin{theorem}\label{thm: euler char of exploring}
    Let $(G_1,\g_1),\ldots, (G_{n},\g_n)$ be a sequence of exploring loops based at a point $z$. Then the Euler characteristic of the filled surface $\Sigma(\g_1,\ldots,\g_n)$ is
    \begin{equation*}
        \chi(\Sigma(\g_1,\ldots,\g_n)) = 1-n.
    \end{equation*}
\end{theorem}
\begin{remark}
    In particular, the process of choosing an $n$-exploring loop terminates at $n = 2g-1$, with a weakly filling system of based geodesic loops.
\end{remark}
\begin{proof}
    We prove the claim by induction on $n$.
    The base case $\ug_1 = \g_1$ yields that $\Sigma(\ug_1)$ is a cylinder, since the systole, $\g_1$, is always simple. Thus $\chi(\Sigma(\ug_1)) = 0 = 1-1$. We now assume that $\chi(\Sigma(\ug_{n-1})) = 2-n$ for $n\geq 2$.

    One has that $\g_n$ intersects $\partial \Sigma(\ug_{n-1})$ exactly twice, c.f.\ Figure \ref{fig: intersectionswithregneighs}. Indeed, as $\g_n$ is exploring and based at $z$, it has to leave $\Sigma(\ug_{n-1})$ and return to it, and thus intersects $\partial \Sigma(\ug_{n-1})$ at least twice.
    As by Theorem \ref{thm: top of len min}, $\g_n$ is simple, in a ball of small enough radius around $z$, say $B$, the loop $\g_n$ forms two radii in $B$, intersecting at $z$. Since by Theorem \ref{thm: len min seq simplicity} $\g_n$ does not intersect any of the loops in $\ug_{n-1}$, one may choose $\eps$ small enough such that
    \begin{equation*}
        \g_n \cap \cN_\eps(\ug_{n-1}) \subseteq B,
    \end{equation*}
    which can only intersect $\partial\cN_\eps(\ug_{n-1})$ exactly twice. Since $\partial \Sigma(\ug_{n-1}) \subseteq \partial \cN_\eps(\ug_{n-1})$, $\g_n$ intersects the boundary exactly twice.
    \begin{figure}[h!]
    \centering
    \begin{subfigure}[t]{0.45\textwidth}
        \centering
        \includegraphics[width=0.8\linewidth]{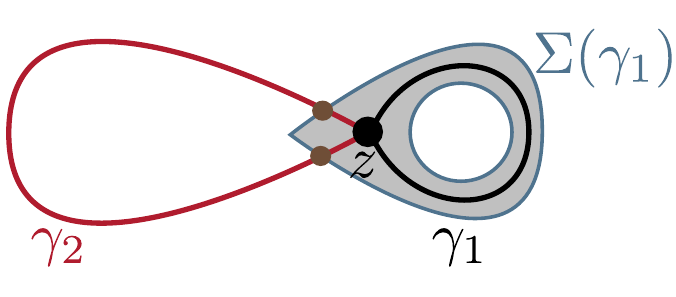}
        \caption{Intersecting one component of\\  $\partial\Sigma(\g_1)$ twice.}
        \label{fig: regneigh1boundaries}
    \end{subfigure}
    ~
    \begin{subfigure}[t]{0.45\textwidth}
        \centering
        \includegraphics[width=0.8\linewidth]{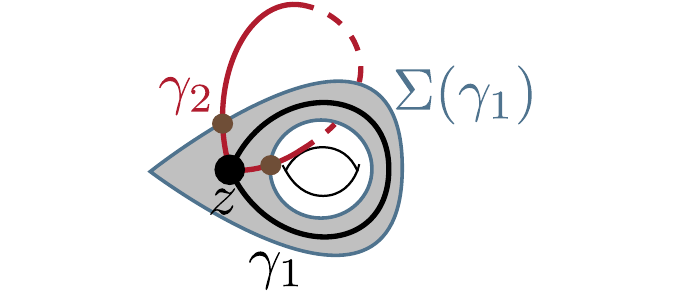}
        \caption{Intersecting two different components of $\partial \Sigma(\g_1)$.}
        \label{fig: regneigh2boundaries}
    \end{subfigure}
    \caption{Illustration of the two different possibilities for intersections of $\g_2$ with $\partial\cN_\eps(\g_1)$. The grey region is $\Sigma(\g_1)$.
    }
    \label{fig: intersectionswithregneighs}
\end{figure}
    
     By the above, the loop $\g_n$ can intersect either one or two boundary components of $\Sigma(\ug_{n-1})$. We note that by Definition \ref{def: weakly filled surface}, if $\g_n$ intersects two different boundary components, then they are not freely homotopic in $X$ (up to orientation). Indeed, the cylinder between them would then be in $\Sigma(\ug_{n-1})$.
    
    As $\g_n$ is simple, the set $\g_n \cap (X-\Sigma(\ug_{n-1}))$, together with the boundary components $\g_n$ intersects, fill a topological pair of pants or once-holed torus $P$ with $P\cap \Sigma(\ug_{n-1}) \subseteq \partial P$, which has $\chi = -1$. By the induction hypothesis and additivity of the Euler characteristic:
    \begin{equation*}
        \chi(\Sigma(\ug_n)) = \chi(\Sigma(\ug_{n-1})) + \chi(P) = 1-n.
    \end{equation*}
\end{proof}
\subsection{The geometry of short loops}\label{subsec: geom of two shortest loops}
In this section we recall two standard results, 
Lemmas \ref{lem: l_ loop collar lemma} and \ref{lem: l_ injrad bound on area around systole}, 
on the geometry around simple geodesic loops on surfaces. We provide a proof of these for the sake of completeness.

We begin with the definition of the injectivity radius based at a point $z$.
\begin{definition}[Injectivity radius]
    Let $\g$ be a systole based at $z$. Then the \emph{injectivity radius} is defined to be
    \begin{equation*}
        \ir\xz:= \frac{\ell(\g)}{2}.
    \end{equation*}
\end{definition}
     \begin{lemma}[Loop Collar Lemma]\label{lem: l_ loop collar lemma}
        Let $\g$ be a systole at $z$. Then any other geodesic loop $\delta$ based at $z$ which is not homotopic w.r.t.\ $z$ to a power of $\g$ has length at least
        \begin{equation*}
            \ell(\delta)\geq 2\ash \prt {\frac{1}{\sinh \frac {\ell(\g)}2}}.
        \end{equation*}
    \end{lemma}
    \begin{proof}
        By \cite[page 192]{Beardon1983}, one has the bound
        \begin{equation*}
            \sinh\prt{\frac{\ell(\g)}{2}}\sinh \prt{\frac{\ell(\delta)}{2}}\geq 1
        \end{equation*}
        which implies the lemma.
    \end{proof}
    \begin{lemma}\label{lem: l_ injrad bound on area around systole}
        Let $r>0$. Then there exists a family of collars contained in $X$, each such collar of area at most $2e^r$, such that if $z\in X$ with $\ir\xz = r$, then $z$ is in one of these collars.
    \end{lemma}
    \begin{proof}\label{prf: proof of injrad bound on area around}
    Let $\ir\xz = r$ and $\g\in \pi_1\xz$ be a systole, i.e.\ $\ell_\xz(\g) = 2r$. The systole $\g$ is freely homotopic to a simple closed geodesic of some length $\ell \leq 2r$ (c.f.\ \cite[Theorem 1.6.6]{buser1992geometry}), which we will denote by $[\g]$.

    By \cite[Theorem 7.35.1]{Beardon1983}, the distance of $z$ to $[\g]$ has the following expression:
    \begin{equation*}
        \cosh d(z,[\g]) = \frac{\sinh r}{\sinh \frac \ell 2} =: \cosh d_{\max}.
    \end{equation*}
    Thus $z$ is in the $d_{\max}$-neighbourhood of $[\g]$ on the surface $X$. The simple closed geodesics in $X$ of length at most $2r$, and their $d_{\max}$-neighbourhoods, form the core geodesics of the family of collars mentioned in the lemma.
    
    The volume of the $d_{\max}$-neighbourhood is less than the volume of the corresponding collar in the cylinder of the closed geodesic $[\g]$. This volume may be computed using Fermi coordinates
    \begin{equation*}
        \int_0^\ell\int_{-d_{\max}}^{d_{\max}}\cosh \rho \d\rho \td t = 2\ell\sinh d_{\max}.
    \end{equation*}
    As $\sinh d_{\max} \leq \cosh d_{\max}$ 
    \begin{equation*}
        2\ell\sinh d_{\max} \leq 2\ell \cosh d_{\max} = \frac{2\ell \sinh r}{\sinh \frac{\ell}{2}}.
    \end{equation*}
    As $\frac{x}{\sinh \frac{x}{2}} \leq 2$ and $\sinh x \leq e^x/2$, the volume of the $d_{\max}$-neighbourhood, which $z$ is in, is at most $2e^r$, proving the lemma.
\end{proof}
\section{Asymptotics of the off-diagonal terms}\label{sec: offdiagterms}
As a few sections have gone by, we remind the reader of the goal of this final section, which is to prove Proposition \ref{prop: off diagonal term claim}. Namely, that for fixed $\tau>0, f$ and $L\geq 1$, the asymptotic
\begin{equation*}
    \Eg\left[ \od \xz\right] = O_{f,L,\tau} \prt{\frac 1{g^2}},
\end{equation*}
holds, where
\begin{equation*}
    \od \xz  = \sum_{\gamma_0\neq \gamma_1^{\pm 1}\in \cP}\sum_{n,m\geq 1}k(\ell_\xz(\g^n))k(\ell_\xz(\g^m)).
\end{equation*}

We proceed to define for each $\xz\in \Omega_g$ the following two sets
\begin{equation}
    \cG_L\xz := \{\g\in \pi_1\xz: \g \text{ primitive and } \ell_\xz(\g)\leq L \}
\end{equation}
and
\begin{equation*}
    N_L\xz := \frac 12 \# \cG_L\xz,
\end{equation*}
where $\cG_L\xz$ is the set of primitive geodesic loops of length $\leq L$, and $N_L\xz$ counts their number, forgetting their orientation. By the compact support of $k,\; \supp k \subseteq [0,L]$, the double sum in the off-diagonal is over pairs of elements in $\cG_L\xz$.

We prove the following simple lemma, which forms one of the reasons for the definition of second shortest primitive loop. Namely, if there are at least two different short primitive loops, then the systole and second shortest primitive loop must be two of them.
\begin{lemma}\label{lem: systole and second shortest in}
    Let $\g_1$ be a systole, and $\g_2$ be a second shortest primitive loop with respect to it, based at a point $z\in X$.
    If $N_L\xz \geq 2$, then 
    \begin{equation*}
        \g_1,\g_2\in \cG_L\xz.
    \end{equation*}
\end{lemma}
\begin{proof}
    Since $\g_1$ is a minimal-length primitive geodesic loop, $\g_1\in N_L\xz$. Since $N_L\xz\geq 2$, there exists some primitive $\delta \in \cG_L\xz$ with $\delta \neq \g_1^{\pm 1}$. Thus $\delta \not\in \{\g_1^m \}_{m\in \Z}$, and hence by definition 
    \begin{equation*}
        \ell(\g_2) \leq \ell(\delta) \leq L,
    \end{equation*}
    which implies that $\g_2\in \cG_L\xz$.
\end{proof}

Define
\begin{equation}\label{eq: ndef of XL}
    X_L:= \{z\in X: N_L\xz \geq 2 \},
\end{equation}
i.e.\ the set of points with at least two distinct primitive geodesic loops.
This set is of fundamental importance, and its properties are expanded upon in the next section.
\subsection{Properties of $X_L$} 
Before endeavouring to prove geometric properties of points in $X_L$, we will need the following standard estimate.
\begin{lemma}\label{lem: nmaximum lattice points}
    Let $X$ be a hyperbolic surface, $z\in X$, and $A>0$ some constant. Then
    \begin{equation*}
        |\{\gamma\in \pi_1\xz: \ell_\xz(\g)\leq A \} | \leq \frac{\sinh^2\frac{A+\ir\xz}{2}}{\sinh^2 \frac{\ir\xz}{2}}.
    \end{equation*}
\end{lemma}
\begin{proof}\label{prf: proof of lattice point counting lemma}
    By lifting the surface to $\H$, $\#\{\g\in \pi_1\xz: \ell_\xz(\g) \leq A \}$ is equivalent to counting the size of 
    \begin{equation*}
        \{\g\in\pi_1\xz:d(\tilde z,\tilde\g.\tilde z)\leq A   \},
    \end{equation*}
    where $\tilde z, \tilde \g$ are appropriate lifts.
    
    Recall that a hyperbolic ball of radius $R$ has area
    \begin{equation*}
        4\pi\sinh^2\frac{R}{2}.
    \end{equation*}
    
    Around each point $\tilde\g.\tilde z \in \H$ such that $d(\tilde z,\tilde \g.\tilde z)\leq A$, put a ball of radius $\ir\xz$. By the definition of the injectivity radius, the interiors of these balls are disjoint, and are contained in the ball of radius $A+\ir\xz$ around $z$, by the triangle inequality. Thus 
    \begin{equation*}
        |\{\g\in \pi_1\xz: d(\tilde z,\tilde \g.\tilde z)\leq L \}| \leq \frac{\vol B_{A+\ir\xz}}{\vol B_{\ir\xz}} = \frac{\sinh^2\frac{A+\ir\xz}{2}}{\sinh^2\frac{\ir\xz}{2}}.
    \end{equation*}
    We note that the above expression is monotone decreasing, and converges to $e^A$ as the injectivity radius $ \ra\infty$.
\end{proof}
An important property of points $z\in X_L$, i.e.\ points with two distinct short primitive geodesic loops, is the following uniform \emph{lower bound} on the injectivity radius.
\begin{lemma}[Injectivity Radii on $X_L$]\label{lem: inj radii on XL}
    For a point $z\in X_L$ one has the following explicit lower bound on the injectivity radius in terms of $L$
    \begin{equation*}
        \ir \xz\geq\frac 12 a(L),
    \end{equation*}
    where $a(L):= 2\ash \prt{\frac{1}{\sinh \frac L2}}$.
\end{lemma}
\begin{proof}
    We prove this claim in the contrapositive. If $z\in X$ has 
    \begin{equation*}
        \ir\xz < \frac{1}{2} a(L),
    \end{equation*}
    then by definition, the length of a systole $\g$ based at $z$ is upper bounded by
        \begin{equation*}
            \ell(\g) < a(L).
        \end{equation*} 
        
        As the function $a(x)$ is monotonic decreasing, the loop collar lemma, Lemma \ref{lem: l_ loop collar lemma}, implies that every other loop that is not a power of $\g$ has length at least 
        \begin{equation*}
            a(\ell(\g)) > a(a(L)) = L,
        \end{equation*}
        which implies $N_L\xz \leq 1$ and hence $z\not\in X_L$.
\end{proof}
We may finally prove a result that will allow us to estimate the expectation of the off-diagonal.
\begin{lemma}[Properties of $\od\xz$ on and off $X_L$]\label{lem: n_properties of OD on XL} For any $z\in X$,
        \begin{itemize}
            \item[(i)] if $z\in X_L$, then $\od\xz \ll_{L,f,\tau} 1$;
            \item[(ii)] otherwise, $z\in X-X_L$, and $\od\xz = 0$.
        \end{itemize}
    \end{lemma}
    \begin{proof}
        For the first claim we trivially estimate $\od \xz$. One has by definition that
        \begin{equation*}
            \od \xz = \sum_{\g_0 \neq \g_1^{\pm 1}\in \cG_L\xz}\sum_{n,m\geq 1} k(\ell_\xz (\g_0^n))k(\ell_\xz(\g_1^m)),
        \end{equation*}
        which by the triangle inequality is trivially
        \begin{equation*}
            \od\xz\leq \prt{\sum_{\id\neq\g\in \pi_1\xz }|k(\ell_\xz (\g))| }^2.
        \end{equation*}
        The boundedness properties of $k$, Lemma \ref{lem: props of k}, imply
        \begin{equation*}
            \od\xz\ll_{L,f,\tau} \prt{|\{\g\in \pi_1\xz: \ell_\xz(\g) \leq L \}|}^2,
        \end{equation*}
        which by Lemma \ref{lem: nmaximum lattice points} 
        \begin{equation}\label{eq: brrr}
            \od\xz\leq \prt{\frac{\sinh^2\frac{L+\ir \xz}{2}}{\sinh^2 \frac{\ir\xz}{2}}}^2.
        \end{equation}
        By the injectivity radius lower bound of Lemma \ref{lem: inj radii on XL}, and the fact that (\ref{eq: brrr}) is monotone decreasing in $\ir\xz$, we have the bound
        \begin{equation*}
            \od\xz\leq \prt{\frac{\sinh^2 \frac{L+\frac 12 a(L)}{2}}{\sinh^2 \frac{a(L)}{4}}}^2 \ll_L 1,
        \end{equation*}
        which proves the first claim of the Lemma.
        
        The second claim follows by definition of $\od$ and $X_L$. If $z\in X-X_L$ then $$N_L(X,z)=\frac 12 \cG_L(X,z) \leq 1$$ which implies that the double sum in the definition of $\od\xz$ (\ref{eq: def of D and OD}) is empty.
    \end{proof}
\subsection{Proof of Proposition \ref{prop: off diagonal term claim}}
Before proceeding with the promised proof, we will require just one last ingredient. 

If $\cL > 0$ and $X\in\cMg$ then we let $N_X^{\chi = -1}(\cL)$ be the number of embedded geodesic once-holed tori or pairs of pants of total boundary length at most $\cL$ in $X$. 

The following asymptotic in $g$ for the expected number of such subsurfaces follows from work of Monk and Thomas.
\begin{lemma}[Theorem 5 of \cite{monkthomastanglefree}]\label{lem: nexpectation of oneholed tori or pops}
  \begin{equation*}
      \E_g^\WP\left[N_X^{\chi =-1}(\cL)\right] =  O_\cL\prt{\frac{1}{g}}.
  \end{equation*}  
\end{lemma}
We now proceed to prove the main result of this section.
\begin{proof}[Proof of Proposition \ref{prop: off diagonal term claim}]\label{prf: n proof of main offdiag proposition}
    Recall the definition of the expectation we are to estimate
    \begin{equation*}
        \Eg[\od\xz] = \frac{1}{4\pi(g-1)} \E_g^\WP\left[\int_X \od\xz \d z\right].
    \end{equation*}
    By Lemma \ref{lem: n_properties of OD on XL}
    \begin{equation}\label{eq: whee}
        \frac{1}{4\pi(g-1)}\E_g^\WP \left[\int_X \od\xz \d z\right] \ll_{L,f,\tau} \frac{1}{g}\E_g^\WP \left[\vol X_L\right].
    \end{equation}
    
    Let $z\in X_L$, $\g_1$ be a systole, and $\g_2$ a second shortest primitive loop, both based at $z$, as per Definition \ref{def: second shortest primitive loop}. As $z\in X_L$, by Lemma \ref{lem: systole and second shortest in},
    $$\ell(\g_1),\ell(\gamma_2)\leq L.$$ 
    By Corollary \ref{cor: length-minimising topology}, $\g_2$ is simple and trivially intersects $\g_1$. The geodesic surface weakly filled by $(\g_1,\g_2)$ is a pair of pants or once-holed torus with total boundary length at most $4L$. As by Lemma \ref{lem: l_ injrad bound on area around systole}, $z$ is in some neighbourhood  of area at most $2e^{\ir\xz} \leq 2e^{L/2}$ of the systole $\g_1$, we have the volume estimate
    \begin{equation*}
        \vol X_L \leq 2e^{L/2}N_X^{\chi=-1}(4L).
    \end{equation*}
    This yields
    \begin{align*}
        \E_g^\WP \left[\vol X_L\right] \leq 2e^{L/2}\E_g^\WP [N_X^{\chi=-1}(4L)],
    \end{align*}
    which implies, together with Lemma \ref{lem: nexpectation of oneholed tori or pops}, that
    \begin{equation*}\label{eq: final XL volume bound}
        \E_g^\WP \left[ \vol X_L\right] \ll_L \frac{1}{g}.
    \end{equation*}
    This finally implies, by (\ref{eq: whee}),
    \begin{equation*}
        \Eg[\od\xz] = O_{L,f,\tau}\prt{\frac{1}{g^2}}.
    \end{equation*}
\end{proof}
\appendix
\section{Gaussian behaviour of eigenfunctions}\label{app: berry agrees}
The aim of this appendix is to prove Theorem \ref{thm: random wave}. First, let us recall our setup. For a hyperbolic surface $X$, we let $\{r_j\}_{j = 0}^\infty\subseteq (0,1/2)i\cup [0,\infty)$ be associated to an orthonormal Laplacian eigenbasis $\{ \phi_j\}_{j\geq 0}^\infty$, with
\begin{equation*}
    \Delta \phi_j = \prt{\frac 14+r_j^2}\phi_j.
\end{equation*}
Fix parameters $L\geq 1, \tau>0$ and a function $f$ with compactly supported Fourier transform $\hf\in C_c^\infty(\R)$, and define
\begin{equation*}
    h(r) := f(L(r-\tau))+f(L(r+\tau)).
\end{equation*}
Then our local Weyl law is
\begin{equation*}
    N\xz = \sum_{j=0}^\infty h(r_j) |\phi_j(z)|^2.
\end{equation*}

We now recall some standard facts about Gaussian random variables. Let $(\Omega, \mathcal F, \P)$ be a probability space. We say a real-valued random variable $\xi:\Omega \to \R$ is a \emph{Gaussian with mean $m$ and variance $\sigma^2$} if its pushforward law
\begin{equation*}
    \xi_* \P=\cN(m,\sigma^2)
\end{equation*}
is the normal Gaussian probability measure of mean $m$ and variance $\sigma^2>0$.
\begin{proposition}[Proposition D.16 of \cite{FollmerSchied2025}]
    Any atomless probability space supports sequences of i.i.d.\ Gaussians.
\end{proposition}
We note that the square of a standard normal  Gaussian has a particular distribution known as the $\chi^2$-distribution, which is a certain type of gamma distribution. The square of a centered Gaussian with variance $\sigma^2>0$ is $\frac{1}{\sigma^2}\chi^2$ distributed.

The following theorem proves that, in terms of the variance, our local Weyl law $N\xz$ behaves as though the eigenfunctions were (appropriately normalised) real-valued i.i.d.\ Gaussians, independent of their energies.
\begin{theorem}\label{thm: random wave}
    For any hyperbolic surface $X$, let  $\{\xi_j^X\}_{j=0}^\infty$ be a sequence of i.i.d.\ real Gaussian random variables of mean 0 and variance $\prt{4\pi(g-1)}^{-1}$ on the probability space 
    \begin{equation*}
        \prt{X,\mathcal B, \frac{\vol}{4\pi(g-1)}}, \quad \text{where $\mathcal B$ is the  standard Borel $\sigma$-algebra of $X$}.
    \end{equation*}

    Then for fixed $\tau>0, f$ and any $L\geq 1$, we have, as $g\to\infty$,
    \begin{equation*}
        \Varg\prt{\sum_{j\geq 0}h(r_j) |\xi_j^X(z)|^2} = \Varg\prt{N\xz} + O_{f,\tau}\prt{\frac{1}{L^2g}}+ O_{f,L,\tau}\prt{\frac{1}{g^2}}.
    \end{equation*}
\end{theorem}
\begin{remark}
    Isometries preserve independence and Gaussianity.
\end{remark}
\begin{proof}
    We first set 
\begin{equation*}
    N(X):= \sum_{j\geq 0 } h(r_j).
\end{equation*}
and
\begin{equation*}
    \cN\xz:= \sum_{j=0}^\infty h(r_j)|\xi_j^X(z)|^2.
\end{equation*}

For any hyperbolic surface $X$, we let $\mu =\vol /4\pi(g-1)$. 

Crucially, as $\xi_j^X$ is Gaussian of variance $(4\pi(g-1))^{-1}$, the pushforward law of $|\xi_j^X|^2_* \mu$ is $\frac{1}{4\pi(g-1)}\chi^2$ distributed.

By Eve's Law
\begin{align}\label{eq: exp of var}
    \Varg (\cN\xz) &= \Eg \left[ \Varg \prt{\cN\xz \middle| N(X)}\right]\\ \label{eq: var of exp}& + \Varg \prt{\Eg\left[\cN\xz \middle| N(X) \right]}.
\end{align}
We start by estimating the second summand.

The variance of the expectation (\ref{eq: var of exp}) equals
\begin{align*}
    \text{VE}:=\Varg &\left(\E_g^\WP\left[\int_X \sum_{j\geq 0} h(r_j)|\xi_j^X(z)|^2 \d \mu(z)\middle| N(X)  \right] \right),
\end{align*}
which, as 
$$\int_X |\xi_j^X(z)|^2 \d\mu(z) =\E_z \left[|\xi_j^X(z)|^2\right]= \frac{1}{4\pi(g-1)},$$ implies,
\begin{equation*}
     \text{VE}=\Varg \prt{\frac{1}{4\pi(g-1)} \E_g^\WP \left[N(X)\middle| N(X) \right]}.
\end{equation*}
The conditional expectation of a random variable with respect to itself is the random variable itself. Thus
\begin{equation*}
     \text{VE}=\frac{1}{\prt{4\pi(g-1)}^2}\Varg \prt{N(X)}.
\end{equation*}
As $N(X)$ is independent of the $z$ variable, by work of Rudnick \cite[Theorem 1.1]{rudnick2023goestatisticsmodulispace}, we have the upper bound of
\begin{equation*}
    \text{VE} = O_{f,\tau,L}\prt{\frac{1}{g^2}}.
\end{equation*}

We now estimate the expectation of the variance (\ref{eq: exp of var}). By standard identities, the expectation of the conditional variance equals 
\begin{equation*}
    \text{EV} := \Eg\left[\Varg \prt{\cN\xz \middle| N(X)} \right] = \Eg \left[\prt{\cN\xz - \Eg\left[\cN\xz \middle | N(X)\right]}^2  \right].
\end{equation*}
As $\Eg\left[\cN\xz \middle | N(X)\right] = \frac{1}{4\pi(g-1)}N(X)$ by the above,
\begin{equation*}
    \text{EV}= \Eg \left[\prt{\cN\xz - \frac{1}{4\pi(g-1)} N(X)}^2 \right].
\end{equation*}
Expanding the square yields 
\begin{equation*}
    \text{EV}=\Eg \left[ \sum_{j,k} h(r_j) h(r_k) \prt{\xi_j^X(z)^2 - \frac{1}{4\pi(g-1)}}\prt{\xi_k^X(z)^2 -\frac{1}{4\pi(g-1)} }\right].
\end{equation*}
Since the $\xi_j^X(z)^2$ are independent, and in particular uncorrelated, the cross-terms vanish and yield
\begin{align*}
    \text{EV}=\E_g^\WP&\left[\sum_{j}h(r_j)^2 \int_X \prt{\xi_j^X(z)^2-\frac{1}{4\pi(g-1)}}^2\d\mu \right].
\end{align*}
The inner integral equals $\Var_z\prt{\xi_j(z)^2}$. As each $\xi_j^X(z)^2$ is $\frac {1}{4\pi(g-1)}\chi^2$ distributed, their variance is $\frac{2}{\prt{4\pi(g-1)}^2}$. Thus,
\begin{align*}
    \text{EV}=\E_g^\WP \left[\sum_{j}h(r_j)^2 \cdot \frac{2}{(4\pi(g-1))^2} \right].
\end{align*}
By Selberg's trace formula, Lemma \ref{lem: selberg trace formula},
\begin{align*}
    \text{EV}&= \frac{1}{4\pi(g-1)}\imii h(r)^2\frac{r\tanh \pi r}{2\pi}\td r\\& + \frac{1}{\prt{4\pi(g-1)}^2}\E_g^\WP \left[ \sum_{\substack{\g \text{ prim.}\\ {\text{closed geod.}}}} \sum_{n\geq 1}\frac{\ell(\g)}{\sinh \frac{n\ell(\g)}{2}}\cF \prt{h^2}(n\ell(\gamma))\right].
\end{align*}
As $\cF(h^2)\in C_c^\infty(\R)$, by the aforementioned work of Rudnick, we have 
\begin{equation*}
    \text{EV}=\frac{1}{4\pi(g-1)}\imii h(r)^2\frac{r\tanh \pi r}{2\pi}\td r + O_{f,\tau,L}\prt{\frac{1}{g^2}}.
\end{equation*}
Expanding the first term, by definition of $h$, it equals
\begin{equation*}
    \frac{1}{4\pi(g-1)}\imii\prt{f(L(r-\tau))+f(L(r+\tau))}^2\frac{r\tanh \pi r}{2\pi}\td r.
\end{equation*}
For large $L$, this has the asymptotic
\begin{equation*}
    \frac{1}{4\pi(g-1)} \frac{\tau \tanh\pi \tau}{\pi L} \|f\|_{L^2(\R)}^2 +O_{f,\tau}\prt{\frac{1}{L^2 g}},
\end{equation*}
which yields the statement of the proposition, by Theorem \ref{thm: main theorem on the variance}.
\end{proof}
\printbibliography
\end{document}